\documentclass[reqno, 11pt]{amsart}
\usepackage{amssymb}
\usepackage[nesting]{hyperref}
\usepackage[pdftex]{graphicx}
\usepackage{listings}
\usepackage{multirow}
\usepackage{placeins}
\usepackage{color}
\usepackage{subfigure}
\usepackage{lscape}
\usepackage{dsfont}
\usepackage{amsmath}
\usepackage{amsfonts}
\usepackage{tikz}
\usepackage{verbatim}
\usepackage{accents}
\usepackage{bbm} 
\usepackage{natbib}

\newcommand{\BlackBoxes}{\global\overfullrule5pt}

\BlackBoxes

\newcommand{\R}{\mathbb{R}}
\newcommand{\N}{\mathbb{N}}

\newcommand{\Q}{\mathbb{Q}}

\newcommand{\A}{\mathcal{A}}
\newcommand{\B}{\mathcal{B}}
\newcommand{\BB}{\mathbb{B}}

\newcommand{\E}{\mathbb{E}}

\newcommand{\F}{\mathcal{F}}
\renewcommand{\H}{\mathcal{H}}

\newcommand{\M}{\mathcal{M}}

\renewcommand{\P}{\mathbb{P}}

\newcommand{\Qc}{\mathcal{Q}}
\newcommand{\Qf}{\mathfrak{Q}}

\newcommand{\T}{\mathcal{T}}

\newcommand{\1}{\mathbbm{1}}
\newcommand{\q}{F^{-1}}

\DeclareMathOperator{\argmin}{argmin}

\DeclareMathOperator{\dif}{d}
\DeclareMathOperator{\ES}{ES}
\DeclareMathOperator{\epi}{epi}
\DeclareMathOperator{\esssup}{ess\,sup}

\DeclareMathOperator{\id}{id}

\DeclareMathOperator{\VaR}{VaR}
\newcommand{\ubar}[1]{\underaccent{\bar}{#1}}

\newtheorem{theorem}{Theorem}

\newtheorem{lemma}[theorem]{Lemma}
\newtheorem{proposition}[theorem]{Proposition}
\theoremstyle{definition}
\newtheorem{example}[theorem]{Example}
\newtheorem{remark}[theorem]{Remark}
\newtheorem{definition}[theorem]{Definition}
\newtheorem{assumption}[theorem]{Assumption}

\numberwithin{equation}{section} \numberwithin{theorem}{section}
\def\0{\kern0pt\-\nobreak\hskip0pt\relax}

\makeatletter
\AtBeginDocument{ \def\@serieslogo{ \vbox to\headheight{ \parindent\z@ \fontsize{6}{7\p@}\selectfont
\vss}}}

\def\makeoverbar#1#2#3#4#5#6#7{ \setbox0=\hbox{$\m@th#2\mkern#5mu{{}#3{}}\mkern#6mu$} \setbox1=\null \dimen@=#4\fontdimen8#13 \dimen@=3.5\dimen@
\advance\dimen@ by \ht0 \dimen@=-#7\dimen@ \advance\dimen@ by \wd0
\ht1=\ht0 \dp1=\dp0 \wd1=\dimen@
\dimen@=\fontdimen8#13 \fontdimen8#13=#4\fontdimen8#13
\rlap{\hbox to \wd0{$\m@th\hss#2{\overline{\box1}}\mkern#5mu$}}
\fontdimen8#13=\dimen@}
\def\mylabel#1#2{{\def\@currentlabel{#2}\label{#1}}}
\makeatother
\copyrightinfo{}

\allowdisplaybreaks

\begin{document}
\title{Dynamic reinsurance in discrete time minimizing the insurer's cost of capital}

\author[A. \smash{Glauner}]{Alexander Glauner${}^\ast$}


\thanks{${}^\ast$ Department of Mathematics, Karlsruhe Institute of Technology (KIT), D-76128 Karlsruhe, Germany,\\ \href{mailto:alexander.glauner@kit.edu} {alexander.glauner@kit.edu}}

\begin{abstract}
	In the classical static optimal reinsurance problem, the cost of capital for the insurer's risk exposure determined by a monetary risk measure is minimized over the class of reinsurance treaties represented by increasing Lipschitz retained loss functions. In this paper, we consider a dynamic extension of this reinsurance problem in discrete time which can be viewed as a risk-sensitive Markov Decision Process. The model allows for both insurance claims and premium income to be stochastic and operates with general risk measures and premium principles. We derive the Bellman equation and show the existence of a Markovian optimal reinsurance policy. Under an infinite planning horizon, the model is shown to be contractive and the optimal reinsurance policy to be stationary. The results are illustrated with examples where the optimal policy can be determined explicitly.
\end{abstract}
\maketitle

\makeatletter \providecommand\@dotsep{5} \makeatother

\vspace{0.5cm}
\begin{minipage}{14cm}
{\small
\begin{description}
\item[\rm \textsc{ Key words}]
{\small Optimal reinsurance, risk measure, cost of capital, risk-sensitive Markov Decision Process}
\item[\rm \textsc{AMS subject classifications}] 
{\small 91G05, 91G70, 90C40}

\end{description}
}
\end{minipage}

\section{Introduction}

Reinsurance is an important instrument for insurance companies to reduce their risk exposure. The optimal design of reinsurance contracts has been studied extensively in the actuarial literature for more than half a century. In their pioneering works, \cite{Borch1960} and \citet{Arrow1963} considered the variance of the retained risk and the exponential utility of terminal wealth as optimality criteria. Since then, a wide range of objectives for the optimal use of reinsurance has been discussed. For a comprehensive literature overview we refer the interested reader to Chapter 8 of \citet{Albrecher2017}. 

Developments in the insurance sector's regulatory framework like Solvency II or the Swiss Solvency Test led to a special interest in monetary risk measures like Value-at-Risk and Expected Shortfall as objective functionals. Economically speaking, the target is to minimize the capital requirement or equivalently the cost of capital for the effective risk after reinsurance which is calculated by the respective monetary risk measure. This line of research has been initiated by \citet{CaiTan2007} who optimized the retention levels of stop-loss contracts under the expected premium principle with respect to Value-at-Risk and Expected Shortfall. \citet{Cai2008} extended the same setting to the non-parametric class of increasing convex reinsurance treaties. A further step to general premium principles has been made by \citet{ChiTan2013} who considered the now standard class of increasing Lipschitz reinsurance contracts. Subsequently, more general risk measures were considered. \citet{Cui2013} were the first to study arbitrary distortion risk measures. Similar findings were reached by \citet{Zhuang2016} using their less technical marginal indemnification function approach. Other extensions of the cost of capital minimization problem concerned additional constraints, see e.g.\ \citet{Lo2017} or multidimensional settings induced by a macroeconomic perspective, see \citet{BauerleGlauner2018}.

Research on the cost of capital minimization problem has so far been focused on static single-period models. Other optimality criteria for the choice of reinsurance have however been considered in dynamic setups. \citet{Schael2004} studied the control of an insurer's surplus process in discrete time by means of investment and reinsurance using general parametric contracts. Optimality criteria were maximization of lifetime dividends and minimization of the ruin probability. The continuous-time versions of these problems gained greater attention in the literature. For an overview we refer to \citet{AlbrecherThonhauser2009} and the books \citet{Schmidli2008}, \citet{AzcureMuler2014}. Several authors used Value-at-Risk and Expected Shortfall based solvency constraints in continuous time reinsurance models, see \citet{Chen2010} and \citet{Liu2013} for two early contributions.

The only study of solvency capital requirements or corresponding cost of capital as optimization target in a dynamic setup is the recent paper by \citet{BauerleGlauner2020c}. They minimized the cost of capital for the discrete-time total discounted loss determined by a general spectral risk measure over the class of increasing Lipschitz retained loss functions from which the insurer selects a treaty in every period depending on the current surplus. Since reinsurance
treaties are typically written for one year \citep{Albrecher2017}, only modeling in discrete time is realistic. Continuous time models are typically used when then insurer's surplus is managed by both reinsurance and capital market instruments. They realistically describe the financial market, but with regard to reinsurance they are only a compromise.  

The aim of this paper is to introduce another dynamic extension in discrete time of the static cost of capital minimization problem. We propose a recursive approach where in the terminal period the insurer faces the static problem and in any earlier period calculates the capital requirement taking into account the period's retained loss, the cost of reinsurance, and the future cost of capital. The latter is a random quantity depending on the development of the future surplus. The recursive minimization of risk measures has been studied by \citet{AsienkiewiczJaskiewicz2017} for an abstract Markov Decision Process and by \citet{BaeuerleJaskiewicz2017,BaeuerleJaskiewicz2018} for a dividend and an optimal growth problem specifically using the entropic risk measure. This choice is motivated by recursive utilities studied extensively in the economic literature since the entropic risk measure happens to be the certainty equivalent of an exponential utility. \citet{BauerleGlauner2020b} generalized the recursive approach to axiomatically characterized general monetary risk measures in an abstract Markov decision model. Here, we show that the approach is well-suited for a dynamic cost of capital optimization. 

The paper is structured as follows: In Section \ref{sec:risk_measures} we recall some important facts about risk measures and premium principles. Then, we introduce the dynamic reinsurance model. The recursive cost of capital minimization problem is solved in Section \ref{sec:finite} under a finite planning horizon. We derive a Bellman equation for the optimal cost of capital and show that there exists a Markovian optimal reinsurance policy only requiring that the monetary risk measures have the Fatou property. Under an infinite planning horizon, we additionally need coherence and can then show that the model is contractive and the optimal reinsurance policy stationary. Addressing a criticism by \citet[Sec.\ 8.4]{Albrecher2017}, who question the suitability of cost of capital minimization as a business objective, we show in Section \ref{sec:connection} that the recursive cost of capital minimization is consistent with profit maximization, the primary target of any company. In Section \ref{sec:examples}, we illustrate our results with analytic examples. We specifically consider Value-at-Risk due to its practical relevance with regard to Solvency II.

\section{Risk measures and premium principles}\label{sec:risk_measures}

Let a probability space $(\Omega, \A, \P)$ and a real number $p \in [1,\infty)$ be fixed. With $q \in (1,\infty]$ we denote the conjugate index satisfying $\frac1p + \frac1q=1$ under the convention $\frac1\infty=0$. Henceforth, $L^p=L^p(\Omega, \A, \P)$ denotes the vector space of real-valued random variables which have an integrable $p$-th moment. $L^p_+$ is the subset of non-negative random variables. We follow the convention of the actuarial literature that positive realizations of random variables represent losses and negative ones gains. A \emph{risk measure} is a functional $\rho : L^p \to \bar{\R}$. The notion of a \emph{premium principle} $\pi: L^p_+ \to  \bar{\R}$ is mathematically closely related but the applications are different. While the former determines the necessary solvency capital to bear a risk, the latter gives the price of (re)insuring it. In contrast to general financial risks, insurance risks are typically non-negative. Hence, it suffices to consider premium principles on $L^p_+$. The properties of risk measures discussed in the sequel apply to premium principles analogously.

\begin{definition}\label{def:rm_properties}
	A risk measure $\rho : L^p \to \bar{\R}$ is
	\begin{enumerate}
		\item \emph{law-invariant} if $\rho(X)=\rho(Y)$ for $X,Y$ with the same distribution.
		\item \emph{monotone} if $X\leq Y$ implies $\rho(X) \leq \rho(Y)$.
		\item \emph{translation invariant} if $\rho(X+m)=\rho(X)+m$ for all $m \in \R$.
		\item \emph{normalized} if $\rho(0)=0$.
		\item \emph{finite} if $\rho(L^p) \subseteq \R$.
		\item \emph{positive homogeneous} if $\rho(\lambda X)=\lambda\rho(X)$ for all $\lambda \in \R_+$.
		\item \emph{convex} if $\rho(\lambda X+(1-\lambda)Y)\leq \lambda\rho(X)+(1-\lambda)\rho(Y)$ for $\lambda \in [0,1]$.
		\item \emph{subadditive} if $\rho(X+Y)\leq \rho(X)+\rho(Y)$ for all $X,Y$.
		\item said to have the \emph{Fatou property}, if for every sequence $\{X_n\}_{n \in \N} \subseteq L^p$ with $|X_n| \leq Y$ $\P$-a.s.\ for some $Y \in L^p$ and $X_n \to X$ $\P$-a.s.\ for some $X \in L^p$ it holds
		\[ \liminf_{n \to \infty} \rho(X_n) \geq \rho(X). \]
	\end{enumerate}
\end{definition}

Throughout, we only consider law-invariant risk measures and premium principles. A risk measure is called \emph{monetary} if it is monotone and translation invariant. It appears to be consensus in the literature that these two properties are a necessary minimal requirement for any risk measure. However, the attribute monetary is rather unusual for premium principles since most of them are monotone but often not translation invariant. Monetary risk measures which are additionally positive homogeneous and subadditive are referred to as \emph{coherent}. Note that positive homogeneity implies normalization and makes convexity and subadditivity equivalent. The Fatou property means that the risk measure is lower semicontinuous w.r.t.\ dominated convergence. The following result can be found in \citet{Rueschendorf2013} as Theorem 7.24.

\begin{lemma}\label{thm:finite_convex_fatou}
	Finite and convex monetary risk measures have the Fatou property.
\end{lemma}

\citet{Pichler2013} showed that coherent risk measures satisfy a triangular inequality.

\begin{lemma}\label{thm:coherent_triangular}
	For a coherent risk measure $\rho$ and $X,Y \in L^p$ it holds
	\[ \left\lvert \rho(X) - \rho(Y) \right\rvert \leq \rho(|X-Y|). \]
\end{lemma}

We denote by $\M_1(\Omega,\A,\P)$ the set of probability measures on $(\Omega,\A)$ which are absolutely continuous with respect to $\P$ and define
\[ \M_1^q(\Omega,\A,\P) = \left\{ \Q \in\M_1(\Omega,\A,\P): \frac{\dif \Q}{\dif \P} \in L^{q}(\Omega,\A,\P) \right\}. \]
Recall that an extended real-valued convex functional is called \emph{proper} if it never attains $-\infty$ and is strictly smaller than $+\infty$ in at least one point. Coherent risk measures have the following dual or robust representation, cf.\ Theorem 7.20 in \citet{Rueschendorf2013}. 

\begin{proposition}\label{thm:coherent_risk_measure_dual} 
	A functional $\rho: L^p \to \bar \R$ is a proper coherent risk measure with the Fatou property if and only if there exists a subset $\Qc \subseteq \M_1^q(\Omega,\A,\P)$ such that
	\begin{align*}
	\rho(X)= \sup_{\Q \in \Qc} \E^\Q[X], \qquad X \in L^p.
	\end{align*}
	The supremum is attained since the subset $\Qc \subseteq \M_1^q(\Omega,\A,\P)$ can be chosen $\sigma(L^q,L^p)$-compact and the functional $\Q \mapsto \E^\Q[X]$ is $\sigma(L^q,L^p)$-continuous.
\end{proposition}

With the dual representation, \cite{BauerleGlauner2020b} derived a complementary inequality to subadditivity.

\begin{lemma}\label{thm:subadditivity_complement}
	A proper coherent risk measure with the Fatou property $\rho: L^p \to \bar \R$ satisfies
	\[ \rho(X+Y) \geq \rho(X) - \rho(-Y) \qquad \text{for all } X,Y \in L^p. \]
\end{lemma}

In the following, $F_X(x)=\P(X\leq x)$ denotes the distribution function, $S_X(x)=1-F_X(x), \ x \in \R,$ the survival function and $\q_X(u)=\inf\{x \in \R: F_X(x)\geq u\}, \ u \in [0,1]$, the quantile function of a random variable $X$. Many established risk measures belong to the large class of distortion risk measures.

\begin{definition}\label{def:dist-rm}
	\begin{enumerate}
		\item An increasing function $g:[0,1] \to [0,1]$ with $g(0)=0$ and $g(1)=1$ is called \emph{distortion function}.
		\item The \emph{distortion risk measure} w.r.t.\ a distortion function $g$ is defined by $\rho_g: L^p \to \bar \R$,
		\[\rho_g(X)= \int_0^\infty g(S_X(x))  \dif x  -  \int_{-\infty}^0 1-g(S_X(x)) \dif x \]
		whenever at least one of the integrals is finite.
		\item The \emph{Wang premium principle} w.r.t.\ a distortion function $g$ is defined by $\pi_g: L^p_+ \to \bar \R$,
		\[\pi_g(X)= (1+\theta) \int_0^\infty g(S_X(x))  \dif x, \quad \theta \geq 0.  \]
	\end{enumerate}
\end{definition}

Distortion risk measures have many of the properties introduced in Definition \ref{def:rm_properties}, see e.g.\ \citet{Sereda2010}.
\begin{lemma}\phantomsection \label{thm:dist-rm_properties}
	\begin{enumerate}
		\item Distortion risk measures are law invariant, monotone, translation invariant, normalized and positive homogeneous. 
		\item A distortion risk measure is subadditive if and only if the distortion function $g$ is concave.
	\end{enumerate}
\end{lemma} 

Wang premium principles share these properties apart from translation invariance. Moreover, they have the Fatou property which has so far not been investigated in the literature.

\begin{lemma}
	For a left-continuous distortion function $g$, the Wang premium principle has the Fatou property.
\end{lemma}
\begin{proof}
	Let $\{X_n\}_{n \in \N} \subseteq L^p_+$ with $X_n \leq Y$ $\P$-a.s.\ for some $Y \in L^p_+$ and $X_n \to X$ $\P$-a.s.\ for $X \in L^p_+$. Especially, $X_n \to X$ in distribution. Therefore, $S_{X_n}(x) \to S_X(x)$ for almost every $x \in R_+$. Since $g$ is left-continuous and increasing it is lower semicontinuous, i.e. $\liminf_{n \to \infty} g(S_{X_n}(x)) \geq g(S_X(x))$ for almost every $x \in R_+$. Finally, Fatou's lemma yields with
	\begin{align*}
	\liminf_{n \to \infty} \pi(X_n) = \liminf_{n \to \infty} (1+\theta) \int_0^{\infty} g(S_{X_n}(x)) \dif x
	\geq (1+\theta) \int_0^{\infty} g(S_X(x)) \dif x = \pi(X).
	\end{align*}
	the assertion.
\end{proof}

There is an alternative representation of distortion risk measures in terms of Lebesgue-Stieltjes integrals based on the quantile function in lieu of the survival function of the risk $X$. 

\begin{remark}\label{rem:spectral-rm}
	For a distortion risk measure $\rho_g$ with left-continuous distortion function $g$ it holds
	\begin{align}\label{eq:dist-rm_Stieltjes}
	\rho_g(X)= \int_0^1 \q_X(u)\dif \bar g(u),
	\end{align}
	where $\bar g(u)=1-g(1-u), \ u \in [0,1],$ is the dual distortion function, cf.\ \citet{Dhaene2012}. For a continuous and  concave distortion function $g:[0,1]\to [0,1]$, the dual distortion function $\bar g: [0,1] \to [0,1]$ is continuous and convex. It can thus be written as $\bar g(x)= \int_0^x \phi(s) \dif s$ for an increasing right-continuous function $\phi:[0,1] \to \R_+$, which is called \emph{spectrum}. By the properties of the Lebesgue-Stieltjes integral, \eqref{eq:dist-rm_Stieltjes} can then be written as 
	\begin{align}\label{eq:spectral_rm}
	\rho_{g}(X)=\rho_{\phi}(X)= \int_0^1 \q_X(u) \phi(u) \dif u.
	\end{align}
	Therefore, distortion risk measures with continuous concave distortion function are referred to as \emph{spectral risk measures}. Note that the continuity of $g$ is an additional requirement only in $0$, since an increasing concave function on $[0,1]$ is already continuous on $(0,1]$.
\end{remark}

Due to Hölder's inequality, spectral risk measures $\rho_\phi:L^p\to \bar \R$ with spectrum $\phi \in L^q$ fulfill
\begin{align*}
\left \lvert \rho_{\phi}(X) \right \rvert&=\left \lvert \int_0^1 \q_X(u) \phi(u) \dif u \right \rvert \leq \int_0^1 |\q_X(u)| \phi(u) \dif u 
= \big(\E|\q_X(U)|^p\big)^{\frac1p}   \big(\E|\phi(U)|^q\big)^{\frac1q} < \infty,
\end{align*}
where $U \sim \mathcal{U}([0,1])$ is arbitrary. Hence, they have the Fatou property by Lemma \ref{thm:finite_convex_fatou}. 

\begin{example}
	\begin{enumerate}
		\item The most widely used risk measure in finance and insurance \emph{Value-at-Risk}
		\[ \VaR_{\alpha}(X) = \q_X(\alpha), \qquad \alpha \in (0,1), \]
		is a distortion risk measure with distortion function $g(u)=\1_{(1-\alpha,1]}(u)$. Since the distortion function is not concave, Value-at-Risk is not coherent and especially not spectral. \citet{BauerleGlauner2020b} have shown that $\VaR$ has the Fatou property.
		\item The lack of coherence can be overcome by using \emph{Expected Shortfall}
		\[ \ES_{\alpha}(X)= \frac{1}{1-\alpha}\int_{\alpha}^1 \q_X(u) \dif u, \qquad \alpha \in [0,1). \]
		The corresponding distortion function $g(u)=	\min\{\frac{u}{1-\alpha},1\}$ is concave and $\ES$ thus coherent. It is also spectral with $\phi(u)=\frac{1}{1-\alpha}\1_{[\alpha,1]}(u)$. Due to the bounded spectrum, $\ES$ has the Fatou property. 
		\item The \emph{Proportional Hazard (PH) premium principle} 
		\[ \pi(X)= (1+\theta) \int_0^\infty S_X(x)^{\gamma}  \dif x, \qquad \theta \geq 0,\ \gamma \in (0,1],  \]
		is an example from the class of Wang premium principles. Note that the distortion function $g(x)=x^{\gamma}, \ \gamma \in (0,1]$ is continuous and concave. For $\gamma=1$, the widely used \emph{Expected premium principle} 
		\[ \pi(X) = (1+\theta)\E[X], \qquad \theta \geq 0, \]
		is a special case.
		\item The \emph{entropic risk measure}
		\[ \rho_\gamma(X)= \frac{1}{\gamma} \log \E\left[ e^{\gamma X}  \right], \qquad \gamma >0, \]
		is also known as \emph{exponential premium principle}. It is a law-invariant and convex monetary risk measure which does not belong to the distortion class. For random variables with existing moment-generating function it has the Fatou property directly by dominated convergence.
	\end{enumerate}
\end{example}

\section{Dynamic reinsurance model}\label{sec:model}

The aim of this paper is to introduce and solve a dynamic extension of the static optimal reinsurance problem
\begin{align}\label{eq:classical_reinsurance_problem}
\min_{f \in \F} \quad r_{\text{CoC}} \cdot \rho \big(f(Y) + \pi_{R}(f)\big),
\end{align}
which has been studied extensively in the literature starting with \citet{CaiTan2007} and generalizations i.a.\ by \citet{ChiTan2013}, \citet{Cui2013}, \citet{Lo2017} and \citet{BauerleGlauner2018}. In this setting, an insurance company incurs a loss $Y \in L^p_+$  at the end of a fixed period due to insurance claims. In order to reduce its risk, the insurer may cede a portion of it to a reinsurance company and retain only $f(Y)$. Here, the reinsurance treaty $f$ determines the retained loss $f(Y(\omega))$ in each scenario $\omega \in \Omega$. For the risk transfer, the insurer has to compensate the reinsurer with a reinsurance premium $\pi_R(f)= \pi_R(Y-f(Y))$ determined by a premium principle $\pi_R:L^p_{+} \to \bar \R$. In order to preclude moral hazard, it is standard in the actuarial literature to assume that both $f$ and the ceded loss function $\id_{\R_+} -f$ are increasing meaning that both the insurer and the reinsurer suffer from higher claims. Otherwise, the insurer might have an incentive to misreport losses or accept unjustified claims. Hence, the set of admissible retained loss functions is
\[ \F= \{f: \R_+ \to \R_+ \mid f(t) \leq t \ \forall t \in \R_+, \ f \text{ increasing}, \ \id_{\R_+}-f \text{ increasing} \}. \]
The insurer's target is to minimize its cost of solvency capital which is calculated as the cost of capital rate $r_{\text{CoC}} \in (0,1]$ times the solvency capital requirement determined by applying the risk measure $\rho$ to the insurer's effective risk after reinsurance.

It is natural to model a dynamic extension of \eqref{eq:classical_reinsurance_problem} in discrete time since reinsurance treaties are typically written for one year \citep{Albrecher2017} and we will focus on the management of the insurer's surplus by means of reinsurance neglecting the possible use of capital market instruments.

In our model, the insurer is endowed with an initial capital $x \in \R$. At the end of each period $[n,n+1), \ n \in \N_0$, he incurs aggregate claims $Y_{n+1} \in L^p_+$ for that period and receives the total premium income $Z_{n+1} \in L^\infty_+$ for the next period. Both quantities are allowed to be stochastic and $(Y_n,Z_n)_{n \in \N}$ is assumed to be an independent sequence of random vectors defined on a common probability space $(\Omega, \A, \P)$.  Requiring that the aggregate losses are independent and fulfill some integrability condition is standard in actuarial science. Often, the premium income is assumed to be deterministic. Here, we allow for some uncertainty or fluctuation but one will at least know an upper bound (complete and timely payment by all insurants).  The insurer's uncontrolled surplus process is given recursively by
\[ X_0=x, \qquad X_{n+1} = X_n - Y_{n+1} + Z_{n+1}. \]
Note that a negative surplus is possible. In order to reduce the downside risk of its surplus process, the insurance company can underwrite a reinsurance treaty represented by a retained loss function $f_n \in \F$ at the beginning of each period $[n,n+1)$. When purchasing reinsurance $f_n$ at time $n$, the insurance company retains the portion $f_n(Y_{n+1})$ of the claims $Y_{n+1}$ arriving at time $n+1$ and the reinsurer covers $Y_{n+1}-f_n(Y_{n+1})$. In return, the insurer has to pay the reinsurance premium  $\pi_{R,n}(f_n)=\pi_{R,n}\big(Y_{n+1}-f_n(Y_{n+1})\big)$. Throughout, we require the premium principles to have the following standard properties.

\begin{assumption}\label{ass:premium}
	$\pi_{R,n}:L^p_+ \to \bar \R$  is a law-invariant, monotone and normalized premium principle with the Fatou property satisfying $\pi_{R,n}(Y_{n+1}) < \infty$ for all $n \in \N_0$.
\end{assumption}
The condition $\pi_{R,n}(Y_{n+1}) <\infty$ means that the risk can be fully ceded at each stage which is natural for a model with a passive reinsurer.

Additionally, we may impose a budget constraint such that for a current capital $x \in \R$ the available reinsurance contracts at time $n$ are 
\[ D_n(x) = \{f \in \F: \pi_{R,n}(f) \leq x^+ \}. \]
This precludes the insurer from purchasing reinsurance on credit. 

For $n \in \N_0$ we denote by $\mathcal{H}_n$ the set of \emph{feasible histories} of the surplus process up to time $n$
\begin{align*}
	h_n = \begin{cases}
		x_0, & \text{if } n=0,\\
		(x_0,f_0,x_1, \dots,f_{n-1},x_n), & \text{if } n \geq 1,
	\end{cases}
\end{align*}
where $f_k \in D_k(x_k)$ for $k \in \N_0$. The decision making at the beginning of each period has to be based on the information available at that time. I.e. the decisions must be functions of the history of the surplus process. 
\begin{definition}
	\begin{enumerate}
		\item A measurable mapping $d_n: \mathcal{H}_n \to A$ with $d_n(h_n) \in D_n(x_n)$ for every $h_n \in \mathcal{H}_n$ is called  \emph{decision rule} at time $n$. A finite sequence $\pi=(d_0, \dots,d_{N-1})$ is called \emph{$N$-stage policy} and a sequence $\pi=(d_0, d_1, \dots)$ is called \emph{policy}.
		\item A decision rule at time $n$ is called \emph{Markov} if it  depends on the current state only, i.e.\ $d_n(h_n)=d_n(x_n)$ for all $h_n \in \mathcal{H}_n$. If all decision rules are Markov, the ($N$-stage) policy is called \emph{Markov}.
		\item An ($N$-stage) policy $\pi$ is called \emph{stationary} if $\pi=(d, \dots,d)$ or $\pi=(d,d,\dots)$, respectively, for some Markov decision rule $d$.
	\end{enumerate}
\end{definition}
With $\Pi \supseteq \Pi^M \supseteq \Pi^S$ we denote the sets of all policies, Markov policies and stationary policies. It will be clear from the context if $N$-stage or infinite stage policies are meant. An admissible policy always exists since full retention is feasible in any scenario. The dynamic of the controlled surplus process under a policy $\pi \in \Pi$ is given by
\begin{align}\label{eq:controlled_surplus}
X_0^\pi=x, \qquad X_{n+1}^\pi = X_n^\pi - d_n(h_n)(Y_{n+1}) - \pi_{R,n}(d_n(h_n)) + Z_{n+1}.
\end{align} 

This setting defines a non-stationary Markov Decision Process (MDP) with the following data:
\begin{itemize}
	\item The state space is the real line $\R$ with Borel $\sigma$-algebra $\B(\R)$.
	\item The action space is $\F$ with Borel $\sigma$-algebra $\B(\F)$. The topology is given in Lemma \ref{thm:model_properties} below.
	\item The independent disturbances are $(Y_n,Z_n)_{n \in \N}$.
	\item The transition function $T_n: \R \times \F \times \R_+ \times \R_+ \to \R$ at time $n$ is given by
	\[  T_n(x,f,y,z) = x - f(y) - \pi_{R,n}(f) + z. \]
	\item Regarding the admissible actions $D_n(x)$ at time $n$ in state $x \in \R$ we have two cases:
	\begin{itemize}
		\item[] Unconstrained: $D_n(x) = \F$ for all $x \in \R$.
		\item[] Budget-constrained: $D_n(x) = \{f \in \F: \pi_{R,n}(f) \leq x^+ \}$ for all $x \in \R$. 
	\end{itemize}
	The set of admissible state-action combinations is $D_n=\{(x,f) \in \R\times\F: f \in D_n(x) \}$. It contains the graph of the constant measurable map $\R \ni x \mapsto \id_{\R_+}$.
	\item The one-stage cost function $c:D \times \R \to \R$ is given by $c(x,f,x')= -x'$, where $x'$ denotes the next state of the surplus process, see Section \ref{sec:finite}.
\end{itemize}

The next lemma summarizes some properties of the dynamic reinsurance model which will be relevant in the following sections.
\begin{lemma}\phantomsection\label{thm:model_properties}
	\begin{enumerate}
		\item The retained loss functions $f \in \F$ are Lipschitz continuous with constant $L\leq1$. Moreover, $\F$ is a Borel space as a compact subset of the metric space $(C(\R_+),m)$ of continuous real-valued functions on $\R_+$ with the metric of compact convergence
		\[ m(f_1,f_2) = \sum_{j=1}^{\infty} 2^{-j} \frac{\max_{0 \leq t \leq j} |f_1(t)-f_2(t)|}{1+ \max_{0 \leq t \leq j} |f_1(t)-f_2(t)| }.\]
		\item The functional $\pi_{R,n}:\F \to \R_+, \ f \mapsto \pi_{R,n}(f)$ is lower semicontinuous.
		\item $D_n(x)$ is a compact subset of $\F$ for all $x \in \R$ and the set-valued mapping $\R \ni x \to D_n(x)$ is upper semicontinuous, i.e.\ if $x_k\to x$ and $a_k\in D_n(x_k)$, $k\in\N$, then $\{a_k\}_{k \in \N}$ has an accumulation point in $D_n(x)$.
		\item The transition function $T_n$ is upper semicontinuous and the one-stage cost $D_n \ni (x,f)\mapsto c(x,f,T_n(x,f,y,z)) = - T_n(x,f,y,z)$ is lower semicontinuous.
	\end{enumerate}
\end{lemma}
\begin{proof}
	\begin{enumerate}
		\item Let $f \in \F$. Since $\id_{\R_+}-f$ is increasing, it holds for $0\leq x\leq y$ that $x-f(x) \leq y-f(y)$. Rearranging and using that $f$ is increasing, too, yields with $|f(x)-f(y)|=f(y)-f(x) \leq y-x = |x-y|$ the Lipschitz continuity with common constant $L=1$. Moreover, $\F$ is pointwise bounded by $\id_{\R_+}$ and closed under pointwise convergence. Hence, $(\F,m)$ is a compact metric space by the Arzelà-Ascoli theorem and as such also complete and separable, i.e. a Borel space.
		\item Let $\{f_k\}_{k \in \N}$ be a sequence in $\F$ such that $f_k \to f \in \F$. Especially, it holds $f_k(x) \to f(x)$ for all $x \in \R_+$ and $Y-f_k(Y) \to Y-f(Y)$ $\P$-a.s. Since $Y-f_n(Y) \leq Y \in L^1$ for all $k \in \N$, the Fatou property of $\pi_{R,n}$ implies
		\[ \liminf_{k \to \infty} \pi_{R,n}(f_k) = \liminf_{k \to \infty} \pi_{R,n}\big(Y-f_k(Y)\big) \geq \pi_{R,n}\big(Y-f(Y)\big)=\pi_{R,n}(f). \]
		\item Due to a), we only have to consider the budget-constrained case. Since $\F$ is compact it suffices to show that $D_n(x) = \{f \in \F: \pi_{R,n}(f) \leq (x)^+ \}$ is closed. This is the case since $D_n(x)$ is a sublevel set of the lower semicontinuous function $\pi_{R,n}:\F \to \R_+$. Furthermore, we show that $D_n$ is closed to obtain the upper semicontinuity from Lemma A.2.2 in \citet{BaeuerleRieder2011}. From the lower semicontinuity of $\pi_{R,n}$ it follows that the epigraph 
		\[ \epi(\pi_{R,n})= \{ (f,x) \in \F\times \R_+: \pi_{R,n}(f) \leq x \} \]
		is closed. Thus, $D_n=\{(x,f): (f,x) \in \epi(\pi_{R,n})\} \cup ( \R_- \times D_n(0))$ is closed, too.
		\item We show that the mapping $\F \times \R_+ \ni (f,y) \mapsto f(y)$ is continuous. Then, the transition function $T_n$ is upper semicontinuous as a sum of upper semicontinuous functions due to part b) and the one-stage cost $c(x,f,T_n(x,f,y,z)) = - T_n(x,f,y,z)$ is lower semicontinuous. Let $\{(f_k,y_k)\}_{k \in \N}$ be a convergent sequence in $\F \times \R_+$ with limit $(f,y)$. Since convergence w.r.t.\ the metric $m$ implies pointwise convergence and all $f_k$ have the Lipschitz constant $L=1$, it follows
		\begin{align*}
			\left| f_k(y_k)-f(y) \right| \leq \left| f_k(y_k)-f_k(y)\right| + \left|f_k(y)-f(y) \right|\leq \left| y_k-y\right| + \left|f_k(y)-f(y) \right|  \to 0.
		\end{align*}
		I.e.\ $\F \times \R_+ \ni (f,y) \mapsto f(y)$ is continuous. \qedhere
	\end{enumerate}
\end{proof}

\section{Cost of capital minimization with finite planning horizon}\label{sec:finite}

Let $N \in \N$ be the finite planning horizon and $\pi=(d_0,\dots,d_{N-1}) \in \Pi$ a policy of the insurance company. At the beginning of the terminal period $[N-1,N)$, the insurer faces the same situation as in the static reinsurance problem \eqref{eq:classical_reinsurance_problem}. The cost of solvency capital is calculated for the period's loss which equals the negative surplus at time $N$:
\[ V_{N-1\pi}(h_{N-1}) = r_{\text{CoC}} \cdot \rho_{N-1}\Big(d_{N-1}(h_{N-1})(Y_N) + \pi_{R,N-1}(d_{N-1}(h_{N-1})) - Z_{N} -x_{N-1} \Big). \]  
In any earlier period $[n,n+1)$, the effective risk relevant for the cost of capital calculation consists of the risk for that period plus the discounted future cost of capital. The latter is a random variable as a measurable function of the next state of the surplus process, see Remark \ref{rem:measurability}. I.e.\ the cost of capital is given by
\begin{align*}
	V_{n\pi}(h_n) = r_{\text{CoC}} \cdot \rho_n\Big(& d_n(h_n)(Y_{n+1}) + \pi_{R,n}(d_n(h_n)) - Z_{n+1} -x_n\\
	 &+ \beta V_{n+1\pi}\big(h_n, \ d_n(h_n), \ x_n+ Z_{n+1} - d_n(h_n)(Y_{n+1}) - \pi_{R,n}(d_n(h_n))  \big) \Big),
\end{align*}
where $\beta \in (0,1]$ is a discount factor. To simplify the notation, we assume that the cost of capital rate $r_{\text{CoC}}$ is included in the discount factor meaning that $\beta$ is of the form
\[ \beta= r_{\text{CoC}} \cdot \frac{1}{1+r}, \]
where $r \in (0,1]$ is the risk-free interest rate per period. Hence, the discount factor still is a quantity in $(0,1]$ and our simplification of the notation entails no restriction. In the first period, one has to multiply once more with the cost of capital rate in order to obtain the overall recursive cost of capital, but for the minimization this is of course not relevant. Hence, we can define the \emph{value of a policy} $\pi=(d_0,\dots,d_{N-1}) \in \Pi$, i.e.\ the cost of capital under this policy, recursively as
\begin{align*}
	V_N(h_N) &=   0,\\
	V_{n\pi}(h_n) &=  \rho_n \Big(  d_n(h_n)(Y_{n+1})+ \pi_{R,n}(d_n(h_n))- Z_{n+1}-x_n \\
	& \phantom{=  \rho_n \Big(}\ + \beta V_{n+1\pi}\big(h_n,\ d_n(h_n),\  x_n + Z_{n+1} - d_n(h_n)(Y_{n+1}) - \pi_{R,n}(d_n(h_n)) \big)\Big).
\end{align*}
The corresponding \emph{value functions} are
\[ V_n(h_n) = \inf_{\pi \in \Pi}  \ V_{n\pi}(h_n), \qquad h_n \in \mathcal{H}_n, \]
and the optimization objective is to determine the optimal recursive cost of solvency capital
\begin{align}\label{eq:opt_crit_finite}
	V_0(x)=\inf_{\pi \in \Pi} \ V_{0\pi}(x), \qquad x \in \R.
\end{align}
For the actuarial interpretation of the optimality criterion one should note the capital requirement for the possible claims is set off against the insurer's capital and income at each stage. Hence, a negative cost of capital at time zero means that a part of the initial capital can be used for other investments or dividend payments without compromising solvency up to the planning horizon. On the other hand, a positive cost of capital represents the opportunity cost of the additional capital needed to make to risk up to the planning horizon acceptable.

We make the following assumption in this section.
\begin{assumption}\phantomsection \label{ass:finite}
	\begin{enumerate}
		\item[(i)] $\rho_0,\dots,\rho_{N-1}: L^p \to \bar \R$ are law-invariant and normalized monetary risk measures with the Fatou property.
		\item[(ii)] It holds $\rho_n(\lambda Y_{n+1}) < \infty$ for all $\lambda \in \R_+$ and $n=0,\dots,N-1$.
	\end{enumerate}
\end{assumption}
I.e.\ we only need a few standard properties for the risk measures plus a finiteness condition to have well-defined policy values. Part (ii) of the assumption reduces to $\rho_n(Y_{n+1})<\infty$ if the risk measure is additionally positive homogeneous.

\begin{remark}\label{rem:measurability}
	For the recursive definition of the policy values to be meaningful, we need to make sure that the risk measures are applied to elements of $L^p(\Omega,\A,\P)$. This has two aspects: integrability will be ensured by Lemma \ref{thm:bounding}, but first of all $V_{n\pi}$ needs to be a measurable function for all $\pi \in \Pi$ and $n=0,\dots,N$. For most risk measures with practical relevance, this is fulfilled. To see this, we proceed by backward induction. For $n=N$ there is noting to show and if $V_{n+1\pi}$ is measurable, the function
	\begin{align*}
		\psi(h_n,y,z) =&\ d_n(h_n)(y)+ \pi_{R,n}(d_n(h_n))- z-x_n \\
		&\ + \beta V_{n+1\pi}\big(h_n,\ d_n(h_n),\  x_n + z - d_n(h_n)(y) - \pi_{R,n}(d_n(h_n)) \big)\Big)
	\end{align*}
	is measurable, too, as a composition of measurable maps. Now let us distinguish different risk measures.
	\begin{itemize}
		\item In the risk-neutral case, i.e.\ for $\rho=\E$, and also for the entropic risk measure $\rho_\gamma$ the measurability of $V_{n\pi}(h_n)= \rho(\psi(h_n,Y_{n+1},Z_{n+1}))$ follows from Fubini's theorem.
		\item For distortion risk measures, the measurability is guaranteed, too. Here, Fubini's theorem yields that the survival function of $\psi(h_n,Y_{n+1},Z_{n+1})$
		\[ S(t|h_n) = \int \1\big\{ \psi(h_n,Y_{n+1}(\omega),Z_{n+1}(\omega)) > t \big\} \P(\dif \omega) \] 
		is measurable. A distortion function $g$ is increasing and hence measurable. So again by Fubini's theorem we obtain the measurability of 
		\begin{align*}
			V_{n\pi}(h_n)= \rho_g(\psi(h_n,Y_{n+1},Z_{n+1})) = \int_0^\infty g(S(t|h_n))  \dif t  -  \int_{-\infty}^0 1-g(S(t|h_n)) \dif t
		\end{align*}
		since the integrands are non-negative and compositions of measurable maps. 
		\item For proper coherent risk measures with the Fatou property one can insert the dual representation of Proposition \ref{thm:coherent_risk_measure_dual}
		\[ V_{n\pi}(h_n)=\sup_{\Q \in \Qc} \E^\Q[\psi(h_n,Y_{n+1},Z_{n+1})]. \] 
		Then, an optimal measurable selection argument as in Theorem 3.6 in \citet{BauerleGlauner2020a} yields the measurability.
	\end{itemize}
	Throughout, it is implicitly assumed that the risk measures are chosen such that all policy values are measurable.
\end{remark}

\begin{lemma}\label{thm:bounding}
	There exist decreasing bounding functions $\ubar b_n:\R \to \R_-$ and $\bar b_n:\R \to \R_+$ such that it holds for all $\pi \in \Pi$ and $n=0,\dots,N-1$
	\[ \ubar b_n(x_n) \leq  V_{n\pi}(h_n) \leq \bar b_n(x_n), \qquad h_n \in \mathcal{H}_n. \]
	The bounding functions are given by
	\[ \ubar b_n(x)= -\ubar c_n -a_n x^+, \qquad \qquad \bar b_n(x)= \bar c_n + a_n x^- \]
	with recursively defined, non-negative coefficients
	\begin{align*}
		\ubar c_{N-1} &= \esssup(Z_N), & \ubar c_n &= (1+\beta a_{n+1}) \esssup(Z_{n+1}) + \beta \ubar c_{n+1}, \\
		\bar c_{N-1} &= \rho_{N-1}(Y_N) + \pi_{R,N-1}(Y_N), & \bar c_n &= \rho_n \big( (1+\beta a_{n+1})Y_{n+1} \big) + (1+\beta a_{n+1}) \pi_{R,n}(Y_{n+1})  + \beta \bar c_{n+1},\\
		a_{N-1}&=1, &  a_n &= 1+ \beta a_{n+1}.
	\end{align*}
\end{lemma}
\begin{proof}
	We proceed by backward induction. At time $N-1$, it follows from the monotonicity, translation invariance and normalization of $\rho_{N-1}$ that for any policy $\pi \in \Pi$
	\begin{align*}
		V_{N-1\pi}(h_{N-1}) &= \rho_{N-1}\Big(d_{N-1}(h_{N-1})(Y_N) + \pi_{R,N-1}(d_{N-1}(h_{N-1})) - Z_{N} -x_{N-1} \Big)\\
		&\geq  \rho_{N-1}\big(-\esssup(Z_N)-x_{N-1} \big)\\
		&= -\esssup(Z_N)-x_{N-1} \\
		&\geq \ubar b_{N-1}(x_{N-1}),\\
		V_{N-1\pi}(h_{N-1}) &= \rho_{N-1}\Big(d_{N-1}(h_{N-1})(Y_N) + \pi_{R,N-1}(d_{N-1}(h_{N-1})) - Z_{N} -x_{N-1} \Big)\\
		&\leq \rho_{N-1}\big(Y_N + \pi_{R,N-1}(Y_N)-x_{N-1}\big)\\
		&= \rho_{N-1}(Y) + \pi_{R,N-1}(Y)-x_{N-1} \\
		&\leq \bar b_{N-1}(x_{N-1}).
	\end{align*}
	Now assume the assertion holds for $n+1$. For any policy $\pi \in \Pi$ it follows from the properties of $\rho_n$ and the monotonicity of $\ubar b_{n+1}$, $\bar b_{n+1}$ that 
	\begin{align*}
		V_{n\pi}(h_n) &\geq \rho_n \Big(  d_n(h_n)(Y_{n+1})+ \pi_{R,n}(d_n(h_n))- Z_{n+1}-x_n\\
		&\phantom{\geq \rho_n \Big(}\ + \beta \ubar b_{n+1}\big( x_n + Z_{n+1} - d_n(h_n)(Y_{n+1}) - \pi_{R,n}(d_n(h_n)) \big)\Big)\\
		&\geq \rho_n\Big(-\esssup(Z_{n+1})-x_n + \beta \ubar b_{n+1}\big(\esssup(Z_{n+1})+x_n\big)\Big)\\
		&\geq -\esssup(Z_{n+1})-x_n^+ + \beta \big( -\ubar c_{n+1} -a_{n+1}(\esssup(Z_{n+1})+x_n^+) \big)\\
		&= \ubar b_n(x_n),\\
		V_{n\pi}(h_n) &\leq \rho_n \Big(  d_n(h_n)(Y_{n+1})+ \pi_{R,n}(d_n(h_n))- Z_{n+1}-x_n\\
		&\phantom{\geq \rho_n \Big(}\ + \beta \bar b_{n+1}\big( x_n + Z_{n+1} - d_n(h_n)(Y_{n+1}) - \pi_{R,n}(d_n(h_n)) \big)\Big)\\
		&\leq \rho_n\Big(Y_{n+1} + \pi_{R,n}(Y_{n+1}) - x_n + \beta \bar b_{n+1}\big(x_n-Y_{n+1}-\pi_{R,n}(Y_{n+1})\big)\Big)\\
		&\leq \rho_n\Big(Y_{n+1} + \pi_{R,n}(Y_{n+1}) + x_n^- + \beta \big(\bar c_{n+1} + a_{n+1}(x_n^- +Y_{n+1}+\pi_{R,n}(Y_{n+1})\big)\Big)\\
		&= \rho_n \Big( (1+\beta a_{n+1})Y_{n+1} \Big) + (1+\beta a_{n+1}) \pi_{R,n}(Y_{n+1})  + \beta \bar c_{n+1} + (1+\beta a_{n+1})x_n^-\\
		&= \bar b_n(x_n). \qedhere
	\end{align*}
\end{proof}

Let us now consider specifically Markov policies $\pi \in \Pi^M$ of the insurance company. Then,
\[ \BB_n = \{v:\R \to \R\mid\ v \text{ decreasing and lower semicontinuous with } \ubar b_n(x) \leq v(x) \leq \bar b_n(x) \ \forall x \in \R \} \]
turns out to be set of potential value functions at time $n$ under such policies. In order to simplify the notation, we define the following operators thereon.
\begin{definition}\label{def:operators}
	For $n=0, \dots, N-1$, $v \in \BB_{n+1}$, $x \in \R$, $f \in D_n(x)$ and a Markov decision rule $d$ let
	\begin{align*}
		L_n v (x,f) &= \rho_n \Big(  f(Y_{n+1}) + \pi_{R,n}(f) - Z_{n+1} -x + \beta v\big(x + Z_{n+1} - f(Y_{n+1}) - \pi_{R,n}(f) \big)\Big),\\
		\T_{n,d} v(x) &= L_n v(x,d(x)), \\
		\T_n  v(x) &= \inf_{f \in D_n(x)}  L_n v (x,f). 
	\end{align*}
\end{definition}
 
Note that the operators are monotone in $v$. Under a Markov policy $\pi=(d_0, \dots, d_{N-1}) \in \Pi^M$, the value iteration can be expressed with the operators. In order to distinguish from the history-dependent case, we denote policy values here with $J$. Setting $J_{N\pi} \equiv 0$, we obtain for $n=0, \dots, N-1$ and $x \in \R$
\begin{align*}
	J_{n\pi}(x) &= \rho_n \Big(  d_n(x)(Y_{n+1}) + \pi_{R,n}(d_n(x)) - Z_{n+1} -x \\
	&\phantom{=\rho_n \Big(}\ + \beta J_{n+1\pi}\big(x + Z_{n+1} - d_n(x)(Y_{n+1}) - \pi_{R,n}(d_n(x)) \big)\Big)= \T_{n d_n} J_{n+1\pi}(x).
\end{align*} 
Let us further define for $n=0, \dots, N-1$ the \emph{Markov value function} 
\begin{align*}
	J_{n}(x) = \inf_{\pi \in \Pi^M} J_{n\pi}(x) , \qquad x \in \R.
\end{align*}

The next result shows that $V_n$ satisfies a Bellman equation and proves that an optimal policy exists and is Markov.
\begin{theorem}\label{thm:finite}
	For $n=0, \dots, N$ the value function $V_n$ only depends on $x_n$, i.e. $V_n(h_n)=J_n(x_n)$ for all $h_n \in \H_n$, lies in $\BB_n$ and satisfies the Bellman equation
	\begin{align*}
		J_N(x) &= 0,\\
		J_n(x) &= \T_n J_{n+1}(x), \quad x \in \R.
	\end{align*}
	Furthermore, for $n= 0, \dots, N-1$ there exist Markov decision rules $d_n^*$ with $\T_{nd_n^*} J_{n+1}=\T_n J_{n+1}$ and every sequence of such minimizers constitutes an optimal policy $\pi^*=(d_0^*,\dots,d_{N-1}^*)$.
\end{theorem}
\begin{proof}
	The proof is by backward induction. At time $N$ we have $V_N=J_N\equiv 0 \in \BB_N$. Assuming the assertion holds at time $n+1$, we have at time $n$:
	\begin{align*}
		V_n(h_n) &= \inf_{\pi \in \Pi} V_{n\pi}(h_n)\\
		&= \inf_{\pi \in \Pi} \rho_n \Big(  d_n(h_n)(Y_{n+1})+ \pi_{R,n}(d_n(h_n))- Z_{n+1}-x_n \\
		& \phantom{= \inf_{\pi \in \Pi} \rho_n \Big(}\ + \beta V_{n+1\pi}\big(h_n,\ d_n(h_n),\  x_n + Z_{n+1} - d_n(h_n)(Y_{n+1}) - \pi_{R,n}(d_n(h_n)) \big)\Big)\\
		&\geq \inf_{\pi \in \Pi} \rho_n \Big(  d_n(h_n)(Y_{n+1})+ \pi_{R,n}(d_n(h_n))- Z_{n+1}-x_n \\
		& \phantom{\geq  \inf_{\pi \in \Pi}\rho_n \Big(}\ + \beta V_{n+1}\big(h_n,\ d_n(h_n),\  x_n + Z_{n+1} - d_n(h_n)(Y_{n+1}) - \pi_{R,n}(d_n(h_n)) \big)\Big)\\
		&= \inf_{\pi \in \Pi} \rho_n \Big(  d_n(h_n)(Y_{n+1})+ \pi_{R,n}(d_n(h_n))- Z_{n+1}-x_n \\
		& \phantom{=  \inf_{\pi \in \Pi} \rho_n \Big(}\ + \beta J_{n+1}\big(x_n + Z_{n+1} - d_n(h_n)(Y_{n+1}) - \pi_{R,n}(d_n(h_n)) \big)\Big)\\
		&= \inf_{f \in D_n(x)}\rho_n \Big(  f(Y_{n+1}) + \pi_{R,n}(f) - Z_{n+1} -x + \beta J_{n+1}\big(x + Z_{n+1} - f(Y_{n+1}) - \pi_{R,n}(f) \big)\Big)
	\end{align*}
	The last equality holds since the minimization does not depend on the entire policy but only on $f=d_n(h_n)$. Here, objective and constraint depend on the history of the process only through $x_n$. Thus, given existence of a minimizing Markov decision rule $d_n^*$, the last line equals $\T_{n d_n^*} J_{n+1}(x_n)$. Again by the induction hypothesis there exists an optimal Markov policy $\pi^* \in \Pi^M$ such that $J_{n+1}=J_{n+1\pi^*}$. Hence, we have
	\begin{align*}
		V_n(h_n) \geq \T_{n d_n^*} J_{n+1}(x_n) = \T_{n d_n^*} J_{n+1\pi^*}(x_n) = J_{n\pi^*}(x_n) \geq J_{n}(x_n) \geq V_n(h_n).
	\end{align*}
	It remains to show the existence of a minimizing Markov decision rule $d_n^*$ and that $J_n \in \BB_n$. We want to apply Proposition 2.4.3 in \citet{BaeuerleRieder2011}. The set-valued mapping $\R \ni x\mapsto D_n(x)$ is compact-valued and upper semicontinuous by Lemma \ref{thm:model_properties} c). Next, we show that $D_n \ni (x,f) \mapsto  L_n v (x,f)$ is lower semicontinuous for every $v \in \BB_{n+1}$. Let $\{(x_k,f_k)\}_{k \in \N}$ be a convergent sequence in $D_n$ with limit $(x^*,f^*) \in D_n$. As a convergent real sequence $\{ x_k \}_{k \in \N}$ is bounded above, by say $\bar x \geq 0$.
	Convergence w.r.t.\ the metric $m$ implies pointwise convergence, i.e. $Y_{n+1} - f_k(Y_{n+1}) \to Y_{n+1}-f(Y_{n+1})$ a.s. By the properties of $\F$ it holds $0 \leq Y_{n+1}-f_k(Y_{n+1}) \leq Y_{n+1} \in L^p$ for all $k \in \N$. Note that the sequence $\{ \inf_{\ell\geq k} \pi_{R,n}(f_\ell)\}_{k \in \N}$ is increasing and bounded from above by $\pi_{R,n}(Y_{n+1})$, i.e. convergent. Now, the Fatou property of $\pi_{R,n}$ implies
	\begin{align}\label{eq:finite_proof1}
		\hat \pi= \lim_{k \to \infty} \inf_{\ell\geq k} \pi_{R,n}(f_\ell) = \liminf_{k \to \infty}\pi_{R,n}(f_k) \geq \pi_{R,n}(f^*).
	\end{align} 
	Further, we have for every $k \in \N$ by the triangular inequality
	\[ \left\lvert x_k+Z_{n+1}-f_k(Y_{n+1})-\inf_{\ell\geq k} \pi_{R,n}(f_\ell) \right\rvert \leq \bar x + Z_{n+1} +Y_{n+1}+\pi_{R,n}(Y_{n+1}) \in L^p. \]
	Again, since convergence w.r.t.\ the metric $m$ implies pointwise convergence, it holds
	\begin{align}\label{eq:eq:finite_proof2}
		x_k+Z_{n+1}-f_k(Y_{n+1})-\inf_{\ell\geq k} \pi_{R,n}(f_\ell) \to x^*+Z_{n+1}-f^*(Y_{n+1})-\hat \pi \quad \text{a.s.\ for } k \to \infty.
	\end{align}
	Lemma \ref{thm:bounding} yields that
	\begin{align*}
		\left\lvert v\big( x_k +Z_{n+1}-f_k(Y_{n+1})-\inf_{\ell\geq k} \pi_{R,n}(f_\ell)  \big) \right\rvert & \leq \max \left\{ \Big| \ubar b_n\big(x_k +Z_{n+1}-f_k(Y_{n+1})-\inf_{\ell\geq k} \pi_{R,n}(f_\ell)\big)\Big|, \right. \notag \\ 
		& \qquad \qquad \left. \Big| \bar b_n\big(x_k +Z_{n+1}-f_k(Y_{n+1})-\inf_{\ell\geq k} \pi_{R,n}(f_\ell)\big) \Big| \right\} \notag \\
		&\leq \max \left\{ \ubar c_n,  \bar c_n\right\} + a_n\big(\bar x +Z_{n+1}+Y_{n+1}+\pi_{R,n}(Y_{n+1})\big),		
	\end{align*}
	which is in $L^p$. The sequence
	\begin{align}\label{eq:eq:finite_proof3}
		\left\{ \inf_{\ell \geq k} v\big( x_\ell + Z_{n+1}(\omega) - f_\ell(Y_{n+1}(\omega))-\inf_{m \geq \ell} \pi_{R,n}(f_m) \big) \right\}_{k \in \N}
	\end{align}
	is increasing and bounded from above, i.e.\ convergent  for almost all $\omega \in \Omega$. Let us denote the almost sure limit by $V \in L^p$. The lower semicontinuity of $v$ implies a.s.
	\begin{align} \label{eq:eq:finite_proof4}
		V&=\lim_{k \to \infty}\inf_{\ell \geq k} v\big( x_\ell + Z_{n+1} - f_\ell(Y_{n+1})-\inf_{m \geq \ell} \pi_{R,n}(f_m) \big) \notag\\
		 &= \liminf_{k \to \infty} v\big( x_\ell + Z_{n+1} - f_\ell(Y_{n+1})-\inf_{m \geq \ell} \pi_{R,n}(f_m) \big) \notag \\
		& \geq v \big( x^*+Z_{n+1}-f^*(Y_{n+1})-\hat \pi \big)\notag\\
		&\geq v \big( x^*+Z_{n+1}-f^*(Y_{n+1})-\pi_{R,n}(f^*) \big).
	\end{align}
	The last inequality holds by \eqref{eq:finite_proof1} since $v$ is decreasing. Now, we get
	\begin{align*} 
		\liminf_{k \to \infty} L_n v(x_k,f_k) & = \liminf_{k \to \infty} \rho_n \Big(  f_k(Y_{n+1}) + \pi_{R,n}(f_k) - Z_{n+1} -x_k\\
		&\phantom{= \liminf_{k \to \infty} \rho_n \Big(}\ + \beta v\big(x_k + Z_{n+1} - f_k(Y_{n+1}) - \pi_{R,n}(f_k) \big)\Big) \\
		& \geq \liminf_{k \to \infty} \rho_n \Big(  f_k(Y_{n+1})+\inf_{\ell\geq k} \pi_{R,n}(f_\ell)-Z_{n+1}-x_k\\
		&\phantom{= \liminf_{k \to \infty} \rho_n \Big(}\ + \beta \inf_{\ell \geq k} v\big( x_\ell + Z_{n+1} - f_\ell(Y_{n+1})-\inf_{m \geq \ell} \pi_{R,n}(f_m) \big)\Big) \\
		&\geq \rho_n \Big(  f^*(Y_{n+1}) + \hat \pi - Z_{n+1} -x^* + \beta V\Big)\\
		& \geq \rho_n \Big( f^*(Y_{n+1}) + \pi_{R,n}(f^*) - Z_{n+1} -x^*\\
		&\phantom{\geq \rho_n \Big(}\ +\beta  v \big( x^*+Z_{n+1}-f^*(Y_{n+1})-\pi_{R,n}(f^*) \big)\Big)\\
		&= L_nv(x^*,f^*).
	\end{align*}
	The first inequality is by the monotonicity of $\rho_n$ and $v$, the second one is by the almost sure convergence of the sequences \eqref{eq:eq:finite_proof2} and \eqref{eq:eq:finite_proof3} together with the Fatou property of $\rho_n$ and the third one is by \eqref{eq:finite_proof1}, \eqref{eq:eq:finite_proof4} together with the monotonicity of $\rho_n$.  
	Hence, $D_n \ni (x,f) \mapsto  L_n v (x,f)$ is lower semicontinuous for every $v \in \BB_{n+1}$. Proposition 2.4.3 in \citet{BaeuerleRieder2011} yields the existence of a minimizing Markov decision rule $d_n^*$ and that $J_n=\T_nJ_{n+1}$ is lower semicontinuous. $J_n$ is also decreasing since the monotonicity of $\rho_n$ and $J_{n+1}$ makes $x \mapsto L_nJ_{n+1}(x,f)$ decreasing for every $f \in D_n(x)$ implying for $x_1\leq x_2$
	\[ J_n(x_1) = \inf_{f \in D_n(x_1)} L_nJ_{n+1}(x_1,f) \geq \inf_{f \in D_n(x_1)} L_nJ_{n+1}(x_2,f) \geq \inf_{f \in D_n(x_2)} L_nJ_{n+1}(x_2,f) = J_n(x_2)  \]
	since $D_n(x_1) \subseteq D_n(x_2)$. Furthermore, $J_n$ is bounded by $\ubar b_n, \bar b_n$ according to Lemma \ref{thm:bounding}, i.e.\ $J_n \in \BB_n$ and the proof is complete. 
\end{proof}

\begin{remark}
	When using general monetary risk measures, we need the claims $Y_1,Y_2,\ldots \in L^p$ to be integrable since our existence results are based on the Fatou property. If however Value-at-Risk is used in each period, we can relax this assumption and consider any real-valued random variables $Y_1,Y_2,\ldots \in L^0$. Especially, heavy-tailed claim size distributions may be used. Firstly, Assumption \ref{ass:finite} (ii) is satisfied, since for $\alpha \in (0,1)$ the quantile $\VaR_{\alpha}(Y)$ is finite for any $Y \in L^0$. And secondly, a recourse on the Fatou property is not needed as Value-at-Risk is lower semicontinuous even w.r.t.\ convergence in distribution, cf.\ the proof of Lemma 2.10 in \citet{BauerleGlauner2020b}. I.e.\ where we have used dominated convergence in the proof of Theorem \ref{thm:finite} we can dispense with the majorants and argue only with almost sure convergence.
\end{remark}

\section{Cost of capital minimization with infinite planning horizon}\label{sec:infinite}

We now consider the optimal reinsurance problem under an infinite planning horizon. This is reasonable if the terminal period is unknown or if one wants to approximate a model with a large but finite planning horizon. Solving the infinite horizon problem will turn out to be easier since it admits a stationary optimal policy. The model is now required to be stationary meaning that the disturbances $(Y_n,Z_n)_{n \in \N}$ are an i.i.d.\ sequence with representative $(Y,Z)$, the premium principle $\pi_R$ and the risk measure $\rho$ do not vary over time and we have strict discounting by $\beta \in (0,1)$. In this section, additional properties are required for the risk measure. The finiteness condition therefore simplifies.

\begin{assumption}\phantomsection \label{ass:infinite}
	\begin{enumerate}
		\item[(i)] $\rho: L^p \to \bar \R$ is law-invariant, proper and coherent risk measure with the Fatou property.
		\item[(ii)] It holds $\rho(Y) < \infty$.
	\end{enumerate}
\end{assumption}

Since the model with infinite planning horizon will be derived as a limit of the one with finite horizon, the consideration can be restricted to Markov policies $\pi=(d_1,d_2,\dots) \in \Pi^M$ due to Theorem \ref{thm:finite}. When calculating limits, it is more convenient to index the value functions with the distance to the time horizon rather than the point in time. This is also referred to as \emph{forward form} of the value iteration and only possible under Markov policies in a stationary model. The value of a policy $\pi=(d_0, d_1\dots ) \in \Pi^M$ up to a planning horizon $N \in \N$ now is
\begin{align}\label{eq:forward_policy_value}
	J_{N\pi}(x)  = \T_{d_0} \circ \dots \circ \T_{d_{N-1}} 0(x) = \T_{d_0} J_{N-1 \vec \pi}(x), \qquad x \in \R,
\end{align}
with $\vec \pi= (d_1,d_2, \dots)$. Hence, it holds $J_{n}^{\text{non-stat}} = J_{N-n}^{\text{stat}}, \ n=0,\dots,N$. Note that the operators from Definition \ref{def:operators} do not depend on the time index in a stationary model. The value function under planning horizon $N \in \N$ is given by
\begin{align*}
	J_N(x) = \inf_{\pi \in \Pi^M} J_{N\pi}(x), \qquad x \in \R.
\end{align*}
By Theorem \ref{thm:finite}, the value function satisfies the Bellman equation
\begin{align}\label{eq:infinite_finite_Bellman}
	J_N(x) =\T J_{N-1}(x) = \T^N 0(x), \qquad x \in \R.
\end{align}

When the planning horizon is infinite, we define the value of a policy $\pi \in \Pi^M$ as 
\begin{align} \label{eq:infinite_policy_value}
	J_{\infty \pi}(x)= \lim_{N \to \infty} J_{N\pi}(x), \qquad x \in \R. 
\end{align}
Hence, the optimality criterion considered in this section is
\begin{align}\label{eq:infinite_value_function}
	J_{\infty}(x) = \inf_{\pi \in \Pi^M} J_{\infty \pi}(x), \qquad x \in \R.
\end{align}

In the stationary model, also the bounding functions no longer depend on the time index.
\begin{lemma}\label{thm:bounding_infinite}
	The decreasing functions $\ubar b:\R \to \R_-$ and $\bar b:\R \to \R_+$ defined by
	\begin{align*}
		\ubar b(x) &= -\frac{x^+}{1-\beta}  - \frac{\ubar \eta}{(1-\beta)^2}, & \ubar \eta&=  \esssup(Z),\\
		\bar b(x) &= \frac{ x^-}{1-\beta} + \frac{\bar \eta}{(1-\beta)^2}, & \bar \eta&=  \rho(Y) + \pi_R(Y),
	\end{align*} 
	satisfy for every planning horizon $N \in \N_0$ and policy $\pi \in \Pi^M$
	\[ \ubar b(x) \leq J_{N\pi}(x) \leq \bar b(x), \qquad x \in \R. \]
\end{lemma}
\begin{proof}
	We proceed by induction. For $N=0$ there is nothing to show. Assuming the assertion holds for $N-1$, it follows
	\begin{align*}
		J_{N\pi}(x)&= \T_{d_0} J_{N-1 \vec \pi}(x) \geq \T_{d_0}\ubar b(x)\\
				&= \rho \Big(  d_0(x)(Y)+ \pi_R(d_0(x))- Z-x + \beta \ubar b \big(x + Z - d_0(x)(Y) - \pi_R(d_0(x)) \big)\Big)\\
				&\geq -\ubar\eta -x + \beta \ubar b(x+\ubar\eta) \geq -\ubar \eta - x^+ + \beta \left( \frac{-\ubar \eta}{(1-\beta)^2} - \frac{x^+ +\ubar \eta}{1-\beta} \right)\\
				&= -\ubar \eta \left(1+ \frac{\beta}{1-\beta} + \frac{\beta}{(1-\beta)^2} \right) -x^+\left(1 + \frac{\beta}{1-\beta} \right) = \ubar b(x).
	\end{align*}
	The first inequality is by the monotonicity of $\T_{d_0}$ and the second one by the monotonicity, translation invariance and normalization of $\rho$. Regarding the upper bound we have
	\begin{align*}
		J_{N\pi}(x)&= \T_{d_0} J_{N-1 \vec \pi}(x) \leq \T_{d_0}\bar b(x)\\
		&= \rho \Big(  d_0(x)(Y)+ \pi_R(d_0(x))- Z-x + \beta \bar b \big(x + Z - d_0(x)(Y) - \pi_R(d_0(x)) \big)\Big)\\
		&\leq \rho \Big( Y+ \pi_R(Y)-x + \beta \bar b \big(x -Y - \pi_R(Y) \big)\Big)\\
		&\leq \rho(Y) + \pi_R(Y) +x^- +\beta \rho\left( \frac{\bar \eta}{(1-\beta)^2} + \frac{x^- + Y + \pi_R(Y)}{1-\beta} \right)\\
		&= \bar \eta \left(1+ \frac{\beta}{1-\beta} + \frac{\beta}{(1-\beta)^2} \right) + x^-\left(1 + \frac{\beta}{1-\beta} \right) = \bar b(x).
	\end{align*}
	The first inequality is again by the monotonicity of $\T_{d_0}$, the second one by the monotonicity of $\rho$ and the third one by subadditivity and translation invariance.
\end{proof}

With analogous arguments as in the proof of Lemma \ref{thm:bounding_infinite} and using the fact that
\begin{align*}
\frac{1}{1-\beta} = 1 + \frac{\beta}{1-\beta} \qquad \text{and} \qquad 1+ \frac{\beta}{1-\beta} + \frac{\beta}{(1-\beta)^2} = \frac{1}{(1-\beta)^2}
\end{align*}
we obtain for $(x,f) \in D$ the two inequalities
\begin{align}
- \rho\big(-\ubar b(x+Z-f(Y)-\pi_R(f)) \big) &\geq \ubar b(x+\ubar \eta) \geq \frac{-\ubar \eta}{(1-\beta)^2} - \frac{x^+ +\ubar \eta}{1-\beta}\notag\\
&= -x^+ \frac{1}{\beta}\left(\frac{1}{1-\beta}-1  \right) -  \ubar \eta \frac{1}{\beta}\left(\frac{1}{(1-\beta)^2}-1 \right)\notag\\
&= \frac{1}{\beta} \left( \ubar b(x) +x^+ +\ubar\eta  \right)\geq \frac{1}{\beta} \left( \ubar b(x) +(1-\beta)x^+ +\ubar\eta  \right)\notag\\
&= \frac{1-(1-\beta)^2}{\beta}\ubar b(x), \label{eq:alpha_down}\\
\rho\big(\bar b(x+Z-f(Y)-\pi_R(f)) \big) &\leq \rho\big(\bar b(x-Y-\pi_R(Y)) \big) \leq \rho\left( \frac{\bar \eta}{(1-\beta)^2} + \frac{Y + \pi_R(Y) + x^-}{1-\beta}\right) \notag\\
&= \frac{x^-}{1-\beta} + \bar \eta \left(\frac{1}{1-\beta} + \frac{1}{(1-\beta)^2} \right)\notag\\
&= x^- \frac{1}{\beta}\left(\frac{1}{1-\beta}-1  \right) +  \bar \eta \frac{1}{\beta}\left(\frac{1}{(1-\beta)^2}-1 \right)\notag \\
&= \frac{1}{\beta} \left( \bar b(x) -x^- -\bar\eta  \right) \leq \frac{1}{\beta} \left( \bar b(x) -(1-\beta)x^- -\bar \eta  \right)\notag\\
&= \frac{1-(1-\beta)^2}{\beta}\bar b(x), \label{eq:alpha_up}
\end{align}
which will be relevant in subsequent proofs.

The next lemma shows that the infinite horizon policy values \eqref{eq:infinite_policy_value} and infinite horizon value function \eqref{eq:infinite_value_function} are well-defined.

\begin{lemma}\label{thm:convergence}
	The sequence $\{J_{N\pi}\}_{N \in \N}$ converges pointwise for every Markov policy $\pi \in \Pi^M$ and the limit function $J_{\infty \pi}$ is lower semicontinuous, decreasing and bounded by $\ubar b,\bar b$.
\end{lemma}
\begin{proof}
	First, we show by induction that for all $N \in \N$
	\begin{align}\label{eq:convergence_1}
	J_{N\pi}(x) \geq J_{N-1 \pi}(x) + (1-(1-\beta)^2)^{N-1} \ubar b(x), \qquad x \in \R.
	\end{align}
	For $N=1$ we have by  Lemma \ref{thm:bounding_infinite}
	\[ J_{1\pi}(x) \geq \ubar b(x) = J_{0\pi}(x) + (1-(1-\beta)^2)^{0} \ubar b(x). \]
	For $N \geq 2$ it follows
	\begin{align*}
	J_{N\pi}(x)&=\rho \Big(  d_0(x)(Y)+ \pi_R(d_0(x))- Z-x + \beta J_{N-1\vec \pi}\big(x + Z - d_0(x)(Y) - \pi_R(d_0(x)) \big)\Big)\\
	&\geq  \rho\Big( d_0(x)(Y)+ \pi_R(d_0(x))- Z-x  + \beta J_{N-2 \vec \pi}\big(x + Z - d_0(x)(Y) - \pi_R(d_0(x))\big) \\
	&\phantom{\geq  \rho\Big(}\ + \beta (1-(1-\beta)^2)^{N-2} \ubar b \big(x + Z - d_0(x)(Y) - \pi_R(d_0(x)) \Big)\\
	&\geq J_{N-1\pi}(x)- \beta (1-(1-\beta)^2)^{N-2} \rho\Big( -\ubar b \big( x + Z - d_0(x)(Y) - \pi_R(d_0(x)) \big) \Big)\\
	&\geq J_{N-1\pi}(x) + (1-(1-\beta)^2)^{N-1} \ubar b(x).
	\end{align*} 
	The first inequality is by the induction hypothesis, the second one is by Lemma \ref{thm:subadditivity_complement} together with the positive homogeneity of $\rho$ and the third one is due to \eqref{eq:alpha_down}. Thus, \eqref{eq:convergence_1} holds. Applying this inequality repeatedly for $N, N-1, \dots, m$ yields
	\begin{align}\label{eq:weak_monotinicity}
	J_{N\pi}(x) &\geq J_{m\pi}(x) + \sum_{k=m}^{N-1} (1-(1-\beta)^2)^k \ubar b(x) \geq  J_{m\pi}(x) +  \delta_m(x),
	\end{align}
	where 
	\[  \delta_m: \R \to (-\infty,0], \quad  \delta_m(x)=   \ubar b(x) \sum_{k=m}^{\infty} (1-(1-\beta)^2)^k, \qquad m \in \N  \]
	are non-positive functions with $\lim_{m \to \infty} \delta_m(x)= 0$ for all $x \in \R$. Hence, the sequence of functions $\{J_{N\pi}\}_{N \in \N}$ is weakly increasing and by Lemma A.1.4 in \citet{BaeuerleRieder2011} convergent to a limit function $J_{\infty \pi}$ which is lower semicontinuous and decreasing. The bounds from Lemma \ref{thm:bounding_infinite} also apply to the limit.
\end{proof}

Since $\ubar b \leq0$ and $\bar b \geq0$, the finite and infinite horizon value functions are bounded in absolute value by
\[ b:\R \to \R_+, \ b(x)= \bar b(x)-\ubar b(x) = \frac{1}{1-\beta} |x| + \frac{1}{(1-\beta)^2} \eta, \qquad \eta=\rho(Y) + \pi_R(Y) + \esssup(Z) \]
and therefore contained in the set 
\[ \BB= \{v:\R \to \R\mid v \text{ lower semicontinuous and decreasing with } \lambda \in \R_+ \text{ s.t. } |v|\leq \lambda b  \}. \]
Endowing it with the weighted supremum norm $\|v\|_b = \sup_{x \in \R} \frac{|v(x)|}{b(x)}$ makes $(\BB, \|\cdot\|_b )$ a complete metric space, cf.\ Proposition 7.2.1 in \citet{HernandezLasserre1999}.

\begin{lemma}\label{thm:contraction}
	The Bellman operator $\T$ is a contraction on $\BB$ with modulus $1-(1-\beta)^2 \in (0,1)$.
\end{lemma}
\begin{proof}
	Let $v \in \BB$. It follows as in the proof of Theorem \ref{thm:finite} that $\T v$ is lower semicontinuous and decreasing. Furthermore, 
	\begin{align*}
	|\T v(x)| &= \inf_{f \in D(x)} \left\lvert \rho \Big(  f(Y) + \pi_R(f) - Z -x +\beta v\big(x + Z - f(Y) - \pi_R(f) \big)\Big) \right\rvert\\
	& \leq \inf_{f \in D(x)} \rho \Big( \left\lvert f(Y) + \pi_R(f) - Z -x \right\rvert \Big) + \beta \rho\Big( \left\lvert v\big(x + Z - f(Y) - \pi_R(f) \big) \right\rvert \Big)\\
	& \leq \inf_{f \in D(x)}  \rho \Big( \left\lvert f(Y) + \pi_R(f) - Z -x \right\rvert \Big) + \beta \rho\Big( \|v\|_b \, b\big(x + Z - f(Y) - \pi_R(f) \big) \Big)\\
	& \leq \rho \Big( Y + \pi_R(Y) + Z +|x| \Big) + \beta \|v\|_b \, \rho\Big( \frac{\eta}{(1-\beta)^2}  + \frac{1}{1-\beta} \big( |x| + Z + Y + \pi_R(Y) \big) \Big)\\
	& \leq \max\{\|v\|_b, 1 \} \left( \eta +|x|  + \beta \Big( \frac{\eta}{(1-\beta)^2}  + \frac{|x| + \eta}{1-\beta}  \Big)  \right)\\
	&= \max\{\|v\|_b, 1 \} \, b(x).
	\end{align*}
	The first inequality is by Lemma \ref{thm:coherent_triangular} and subadditivity. Hence, the operator $\T$ is an endofunction on $\BB$ and it remains to verify the Lipschitz constant $1-(1-\beta)^2$. For $v_1,v_2 \in \BB$ it holds
	\begin{align*}
	\left \lvert \T v_1(x) - \T v_2(x) \right\rvert &\leq \sup_{f \in D(x)} \left\lvert \rho \Big(  f(Y) + \pi_R(f) - Z -x +\beta v_1\big(x + Z - f(Y) - \pi_R(f) \big)\Big)\right.\\
	&\phantom{\leq \sup_{f \in D(x)}}\ \ \left. - \rho \Big(  f(Y) + \pi_R(f) - Z -x +\beta v_2\big(x + Z - f(Y) - \pi_R(f) \big)\Big)\right\rvert\\
	& \leq \beta \sup_{f \in D(x)}  \rho \Big(\left\lvert v_1\big(x + Z - f(Y) - \pi_R(f) \big) - v_2\big(x + Z - f(Y) - \pi_R(f) \big) \right\rvert \Big)\\
	& \leq \beta \|v_1-v_2\|_b \sup_{f \in D(x)}  \rho \Big(b\big(x + Z - f(Y) - \pi_R(f) \big)\Big)\\
	& \leq \beta \|v_1-v_2\|_b \sup_{f \in D(x)}  \Big[ \rho \Big(\bar b\big(x + Z - f(Y) - \pi_R(f) \big)\Big)\\
	&\phantom{\leq \beta \|v_1-v_2\|_b \sup_{f \in D(x)} \Big[}\ + \rho \Big(-\ubar b\big(x + Z - f(Y) - \pi_R(f) \big)\Big)\Big]\\
	&\leq (1-(1-\beta)^2) \|v_1-v_2\|_b [\bar b(x)-\ubar b(x)] = (1-(1-\beta)^2) \|v_1-v_2\|_b b(x).
	\end{align*}
	Dividing by $b(x)$ and taking the supremum over $x \in \R$ on the left hand side completes the proof. Note that the second inequality holds by Lemma \ref{thm:coherent_triangular}, the fourth one due to $b =\bar b-\ubar b$ and the subadditivity of $\rho$ and the last one is by \eqref{eq:alpha_down}, \eqref{eq:alpha_up}.
\end{proof} 

Under a finite planning horizon $N \in \N$ we have characterized the value function with the Bellman equation \eqref{eq:infinite_finite_Bellman}. We will show that this is compatible with the optimality criterion of the infinite horizon model \eqref{eq:infinite_value_function}. To this end, we define the \emph{limit value function}
\[ J(x)= \lim_{N \to \infty} J_N(x), \qquad x \in \R. \]
Note that the limit exists since it follows from \eqref{eq:weak_monotinicity} that $J_N \ge J_m+\delta_m$ for all $N\ge m$ making $J_N$ weakly increasing.

\begin{theorem}\phantomsection\label{thm:infinite}
	\begin{enumerate}
		\item The limit value function $J$ is the unique fixed point of the Bellman operator $\T$ in $\BB$.
		\item There exists a Markov decision rule $d^*$ such that
		\[ \T_{d^*} J (x) = \T J(x) , \qquad x \in \R.  \]
		\item Each stationary reinsurance policy $\pi^*=(d^*,d^*,\dots)$ induced by a Markov decision rule $d^*$ as in part b) is optimal for optimization problem \eqref{eq:infinite_value_function} and it holds $J_{\infty} = J$. 
	\end{enumerate}
\end{theorem}
\begin{proof}
	\begin{enumerate}
		\item The fact that $J$ is the unique fixed point of the operator $\T$ in $\BB$ follows directly from Banach's Fixed Point Theorem using Lemma \ref{thm:contraction}.
		\item The existence of a minimizing Markov decision rule follows from the respective result in the finite horizon case, cf.\ Theorem \ref{thm:finite}. 
		\item Let $d^*$ be a Markov decision rule as in part b) and $\pi^*=(d^*,d^*,\dots)$. Then it holds 
		$$J(x)  \leq J_{\infty}(x) \leq J_{\infty \pi^*}(x), \qquad x \in \R.$$
		The second inequality is by definition. Regarding the first one note that for any $\pi \in \Pi^M$ we have $J_N \leq J_{N\pi}$ for all $N \in \N_0$. Letting $N \to \infty$ yields $J \leq J_{\infty \pi}$. Since $\pi \in \Pi^M$ was arbitrary, we get $ J \leq \inf_{\pi \in \Pi^M} J_{\infty \pi} = J_{\infty}$. It remains to show
		\begin{align}\label{eq:infinite_proof_1}
		J_{\infty \pi^*}(x) \leq J(x), \qquad x \in \R.
		\end{align} 
		To that end, we will prove by induction that for all $N \in \N_0$ and $x \in \R$
		\begin{align}\label{eq:infinite_proof_2}
		J(x) \geq J_{N\pi^*}(x) + (1-(1-\beta)^2)^N \ubar b(x).
		\end{align}
		Letting $N \to \infty$ in \eqref{eq:infinite_proof_2} yields \eqref{eq:infinite_proof_1} and concludes the proof. For $N=0$ equation \eqref{eq:infinite_proof_2} reduces to $J(x) \geq \ubar b(x)$, which holds by Lemma \ref{thm:bounding_infinite}. For $N \geq 1$ the induction hypothesis yields
		\begin{align*}
			\qquad J(x) &= \T_{d^*} J(x) \geq \T_{d^*} \left(J_{N-1\pi^*} + (1-(1-\beta)^2)^{N-1} \ubar b\right)(x) \\
			&\geq \rho \Big(  d^*(x)(Y)+ \pi_R(d^*(x))- Z-x + \beta J_{N-1\pi^*}\big(x + Z - d^*(x)(Y) - \pi_R(d^*(x)) \big)\Big) \\
			&\phantom{\geq}\ - \beta(1-(1-\beta)^2)^{N-1} \rho\Big(- \ubar b \big(x + Z - d^*(x)(Y) - \pi_R(d^*(x)\big) \Big)\\
			&\geq  J_{N\pi^*}(x) + (1-(1-\beta)^2)^{N} \ubar b(x).
		\end{align*}
		The second inequality is by Lemma \ref{thm:subadditivity_complement} together with the positive homogeneity of $\rho$ and the last one is by \eqref{eq:alpha_down}. \qedhere
	\end{enumerate}
\end{proof}

\section{Connection to profit maximization}\label{sec:connection}

In this section, we show that recursive cost of solvency capital minimization is in accordance with the primary target of any insurance company: profit maximization. For notational convenience we consider the stationary model regardless of the planning horizon and use forward indexing for the value functions. Let $\rho$ be a law-invariant, proper and coherent risk measure with the Fatou property satisfying  $\rho(Y) < \infty$. By inserting the dual representation of Proposition \ref{thm:coherent_risk_measure_dual} in the Bellman equation, we get
\begin{align*}
	J_0(x) &= 0,\\
	J_N(x) &= \inf_{f \in D(x)} \sup_{\Q \in \Qc} \E^\Q\Big[ f(Y) + \pi_R(f) - Z -x +\beta J_{N-1}\big(x + Z - f(Y) - \pi_R(f) \big) \Big], \qquad x \in \R,
\end{align*}
i.e.\ the Bellman equation of a distributionally robust Markov Decision Process as considered in \citet{BauerleGlauner2020a}. Through this connection, we can derive a closed form expression for our recursively defined optimality criterion which turns out to be related to profit maximization.

Since the claims $Y_1,Y_2,\dots$ and premium income $Z_1,Z_2,\dots$ are i.i.d.\ we can w.l.o.g.\ assume that the probability space has a product structure 
\[ (\Omega,\A,\P) = \bigotimes_{n=1}^\infty (\Omega_1,\A_1,\P_1) \]
with $(Y_n,Z_n)(\omega)=(Y_n,Z_n)(\omega_n)$ only depending on component $\omega_n$ of $\omega=(\omega_1,\omega_2,\dots) \in \Omega$. Besides, the probability measure $\P_1$ on $(\Omega_1,\A_1)$ can w.l.o.g.\ assumed to be as separable by constructing the random variables canonically since $\B(\R_+^2)$ is countably generated which makes any probability measure on it separable, see \citet[1.12]{Bogachev2007}. Furthermore, note that
\begin{align*}
\rho\big( (x+Z-f(Y)-\pi_R(f))^+ \big) &\leq x^+ + \esssup(Z) \leq -\ubar b(x),\\
\rho\big( (x+Z-f(Y)-\pi_R(f))^- \big) &\leq \rho(Y)+\pi(Y) +x^-\leq \bar b(x)
\end{align*}
holds for all $(x,f) \in D$. This implies that the assumptions of Theorem 6.1 in \cite{BauerleGlauner2020b} are satisfied and we can conclude the following:

\begin{proposition}
	For a Markov policy $\pi=(d_0,d_1,\dots) \in \Pi^M$ and $\gamma=(\gamma_0,\gamma_1,\dots)$, where $\gamma_n:D \to \Qc$ is measurable, we define the transition kernel
	\begin{align*}
	Q_n^{\pi\gamma}(B|x)= \int \1_B\big(x+Z_{n+1}(\omega_{n+1})-d_n(x)(Y_{n+1}(\omega_{n+1}))-\pi_R(d_n(x))\big) \gamma_n(\dif \omega_{n+1}|x,d_n(x)), 
	\end{align*}
	$B \in \B(\R), \ x \in \R,$ and the law of motion $\Q^{\pi\gamma}_x = \delta_x \otimes Q_0^{\pi\gamma} \otimes Q_1^{\pi\gamma} \otimes \dots$ of the surplus process induced by the Theorem of  Ionescu-Tulcea. The set of all possible laws of motion under policy $\pi \in \Pi^M$ is denoted by $\Qf_\pi = \{ \Q_x^{\pi \gamma}: \gamma \in \Gamma \}$ with $\Gamma$ being the set of all possible $\gamma$. Then it holds for $N \in \N \cup \{\infty\}$
	\begin{align}\label{eq:robust}
	J_N(x) = \inf_{\pi \in \Pi^M} \sup_{\Q \in \Qf_\pi} \E^{\Q}\left[ \sum_{n=0}^{N-1} \beta^{n} \big( d_n(x)(Y_{n+1})+\pi_R(d_n(X_n^\pi))-Z_{n+1} -X_n^\pi \big) \right].
	\end{align}
\end{proposition}
\begin{proof}
	The assertion follows directly from Theorem 6.1 in \cite{BauerleGlauner2020b}.
\end{proof}

To interpret this closed form expression for the value functions let us reformulate \eqref{eq:robust} to
\begin{align*} 
J_N(x)&= - \sup_{\pi \in \Pi^M} \inf_{\Q \in \Qf_\pi} \E^{\Q} \left[ \sum_{n=0}^{N-1} \beta^{n} \Big(X_n^\pi + Z_{n+1} -d_n(X_n^\pi)(Y_{n+1}) - \pi_R(d_n(X_n^\pi))\Big) \right]\\ 
&=\sum_{n=0}^{N-1} \beta^{n}x - \sup_{\pi \in \Pi^M} \inf_{\Q \in \Qf_\pi} \E^{\Q} \left[ \sum_{n=0}^{N-1} \Big( \sum_{j=n}^{N-1} \beta^{j} \Big) \Big(Z_{n+1} - d_n(X_n^\pi)(Y_{n+1}) - \pi_R(d_n(X_n^\pi)) \Big) \right]. 
\end{align*}
The second equality can be obtained inductively by inserting the dynamics of the surplus process \eqref{eq:controlled_surplus}. Here, we have a robust maximization of total profit with higher weights on earlier periods. This addresses a fundamental criticism of cost of capital minimization as an optimality criterion for reinsurance design by \citet[Sec.\ 8.4]{Albrecher2017}. The authors state that if the minimization of the cost of capital was the driving criterion of the insurer, it would be optimal in the long run to stay out of business altogether and thereby achieve zero cost of capital. This viewpoint brings the suitability of the recursive optimality criterion into question since it would be applied over several periods. However, under a coherent risk measure the calculations above show that the recursive criterion is indeed in accordance with profit maximization.

\section{Examples}\label{sec:examples}

In this section, we present two examples where the optimal reinsurance policy can be determined analytically. Moreover, we give an example illustrating the difficulties of the general case. The first example studies Value-at-Risk as a concrete choice for the risk measure $\rho$. This has particular practical relevance with regard to Solvency II.

\begin{example}\label{ex:var}
	Let $\rho_n=\VaR_{\alpha_n}$ and the insurer's premium income be deterministic, i.e.\ $Z_{n+1} \equiv z_{n+1} \in \R_+$, $n=0,\dots,N-1$. That is, we focus on the risk in the insurance claims and neglect potential fluctuations of premium payments from the insurer's customers. Using Theorem \ref{thm:finite}, we have to solve the Bellman equation
	\begin{align*}
	J_n(x) &= \inf_{f \in D_n(x)}  \VaR_{\alpha_n} \Big(  f(Y_{n+1}) + \pi_{R,n}(f) -z_{n+1} -x  + \beta J_{n+1}\big( x + z_{n+1} - f(Y_{n+1}) - \pi_{R,n}(f)  \big)\Big)\\
	&= \inf_{f \in D_n(x)}  \VaR_{\alpha_n} \Big( \phi\big( f(Y_{n+1}) + \pi_{R,n}(f) -z_{n+1} -x  \big) \Big)\\
	&= \inf_{f \in D_n(x)}  \phi \Big( \VaR_{\alpha_n}\big( f(Y_{n+1})\big) + \pi_{R,n}(f) -z_{n+1} -x  \Big)
	\end{align*}
	Here, we used that quantiles can be interchanged with the increasing lower semicontinuous (i.e.\ left-continuous) function $\phi(x)=x+\beta J_{n+1}(-x)$, see e.g.\ Proposition 2.2 in \cite{BauerleGlauner2018}. The increasing transformation $\phi$ can be dropped and it remains to solve
	\begin{align}\label{eq:example_var}
	\inf_{f \in D_n(x)} \VaR_{\alpha_n}(f(Y_{n+1}))  + \pi_{R,n}(f).
	\end{align}
	This is a static optimal reinsurance problem with budget constraint. I.e.\ the dynamic reinsurance problem \eqref{eq:opt_crit_finite} possesses a myopic optimal policy under Value-at-Risk.
	
	Extending an approach used in \citet{ChiTan2013} and \citet{BauerleGlauner2018} to problems with constraints, \eqref{eq:example_var} can be reduced to a finite dimensional problem. By repeating the interchange argument, the minimization problem can be rewritten as 
	\[  \inf_{f \in D_n(x)} f\big(\VaR_{\alpha_n}(Y_{n+1})\big)  + \pi_{R,n}(f). \]
	Now define
	\[ h_{a}(x) = \max\left\{ \min\{ a,x\}, \ x - \VaR_{\alpha_n}(Y_{n+1})+ a \right\} , \qquad x \in \R_+, \ 0 \leq a \leq \VaR_{\alpha_n}(Y_{n+1}).  \]
	This is the retained loss function corresponding to a layer reinsurance treaty with deductible $a$ and upper bound $\VaR_{\alpha_n}(Y) - a$.   Clearly, $h_{a} \in \F$ for all $ a \in [0, \VaR_{\alpha_n}(Y_{n+1})]$. Fix $f \in \F$. We write $h_{f}$ short hand for $h_{a}$ when $a=f(\VaR_{\alpha_n}(Y_{n+1}))$. Observe that $f(\VaR_{\alpha_n}(Y_{n+1})) \in  [0, \VaR_{\alpha_n}(Y_{n+1})]$. Simply by inserting we find that
	\begin{align*}
	h_{f}(\VaR_{\alpha_n}(Y_{n+1})) = f(\VaR_{\alpha_n}(Y_{n+1})).
	\end{align*}
	Moreover, it holds $\pi_{R,n}(h_{f}) \leq \pi_{R,n}(f)$. This can be seen as follows. If $0\leq x < f(\VaR_{\alpha_n}(Y_{n+1}))$, then $h_{f}(x) = x \geq f(x)$ as $f$ is bounded by the identity. If $f(\VaR_{\alpha_n}(Y_{n+1})) \leq x < \VaR_{\alpha_n}(Y_{n+1})$, then $h_{f}(x)= f(\VaR_{\alpha_n}(Y_{n+1})) \geq f(x)$ since $f$ is increasing. Finally if $x \geq \VaR_{\alpha_n}(Y_{n+1})$, then $h_f(x) = x - \VaR_{\alpha_n}(Y_{n+1}) + f(\VaR_{\alpha_n}(Y_{n+1})) \geq f(x)$ as $f$ is $1$-Lipschitz, cf. Lemma \ref{thm:model_properties} a). Consequently, $Y_{n+1}-h_{f}(Y_{n+1}) \leq Y_{n+1}-f(Y_{n+1})$ and by monotonicity $\pi_{R,n}(h_{f}) \leq \pi_{R,n}(f)$. 
	I.e.\ $h_{f}$ is weakly better than $f$ with respect to the objective function and satisfies the constraint if $f$ does. Therefore, it suffices to consider the reduced problem
	\begin{align}\label{eq:example_var_red}
	\inf_{0 \leq a \leq \VaR_{\alpha_n}(Y_{n+1})} \ a + \pi_{R,n}(h_{a}) \quad \text{such that} \quad \pi_{R,n}(h_{a}) \leq x^+.
	\end{align}
	
	For solving \eqref{eq:example_var_red} explicitly, one has to specify the premium principle. We consider exemplarily a Wang premium principle
	\[ \pi_{R,n}(X)= (1+\theta) \int_0^{\infty} g(S_X(x))\dif x, \quad \theta \geq 0, \] 
	where we only assume that the distortion function $g$ is left-continuous. This includes any PH premium, especially the expected premium principle. In order to calculate the premium for $h_a$, we have to determine the survival function of $Y_{n+1} - h_{a}(Y_{n+1}) = \min\{ (Y_{n+1}-a)^+, \VaR_{\alpha_n}(Y_{n+1})-a \}$:
	\begin{align*}
	\P(Y_{n+1} - h_{a}(Y_{n+1}) > y) = \begin{cases}
	\P\big((Y_{n+1}-a)_+ >y\big) = S_{Y_{n+1}}(y+a),& 0 \leq y < \VaR_{\alpha_n}(Y_{n+1}) - a,\\
	0, & y \geq \VaR_{\alpha_n}(Y_{n+1}) - a.
	\end{cases}
	\end{align*}
	It follows
	\begin{align*}
	\pi_{R,n}(h_{a}) = (1+\theta) \int_0^{\VaR_{\alpha_n}(Y_{n+1}) - a} g\big( S_{Y_{n+1}}(y+a) \big) \dif y = (1+\theta) \int_a^{\VaR_{\alpha_n}(Y_{n+1}) } g\big( S_{Y_{n+1}}(y) \big) \dif y
	\end{align*}
	The derivative of the objective function $$\psi(a)= a + (1+\theta) \int_a^{\VaR_{\alpha_n}(Y_{n+1}) } g\big( S_{Y_{n+1}}(y) \big) \dif y, \qquad 0 \leq a \leq \VaR_{\alpha_n}(Y_{n+1})$$
	is given by $\psi'(a) = 1- (1+\theta) g(S_{Y_{n+1}}(a))$. Since the distortion function $g$ is left-continuous, $g \circ S_Y$ is itself a survival function. Thus, $\psi'$ is increasing and right continuous. I.e.\ its generalized inverse
	\[ \psi^{'-1}(z) = \inf \{ a \in [ 0, \VaR_{\alpha_n}(Y_{n+1}) ]: \ \psi'(a) \geq z \} \]
	is well-defined for every $z$ in the range of $\psi'$. Let us distinguish two cases:\\
	\emph{Case 1:}  $g(1-\alpha_n) < \frac{1}{1+\theta}$\\
	By definition of a quantile, we have $S_{Y_{n+1}}(\VaR_{\alpha_n}(Y_{n+1})) \leq 1-\alpha_n$. Since $g$ is increasing it follows 
	\begin{align*}
	\psi'(\VaR_{\alpha_n}(Y_{n+1}))  = 1- (1+\theta) g\big( S_{Y_{n+1}}(\VaR_{\alpha_n}(Y_{n+1})) \big) \geq 1- (1+\theta) g( 1-\alpha_n ) >0.
	\end{align*}
	Hence, $\psi$ is strictly increasing on $[\psi'^{-1}(0),\VaR_{\alpha_n}(Y_{n+1})]$.\\
	\emph{Case 2:}  $g(1-\alpha_n) \geq \frac{1}{1+\theta}$\\
	Let $a < \VaR_{\alpha_n}(Y_{n+1})$. Then $S_{Y_{n+1}}(a)>1-\alpha_n$ and as $g$ is increasing
	\begin{align*}
	\psi'(a)  = 1- (1+\theta) g\big( S_{Y_{n+1}}(a) \big) \leq 1- (1+\theta) g( 1-\alpha ) \leq 0.
	\end{align*}
	I.e.\ $\psi$ is decreasing on $[0,\VaR_{\alpha_n}(Y_{n+1})]$.\\
	Note that in practice $\alpha_n$ is chosen very close to $1$ and $\theta$ smaller than $1$, so only the first case is actually relevant. Let us define
	\[ a^*_n=\begin{cases}
	\psi'^{-1}(0), & \text{if } g(1-\alpha_n) < \frac{1}{1+\theta},\\
	\VaR_{\alpha_n}(Y_{n+1}), & \text{otherwise}.
	\end{cases} \]
	Note that $a=\VaR_{\alpha_n}(Y_{n+1})$ is always feasible for optimization problem \eqref{eq:example_var_red} and that $a \mapsto \pi_{R,n}(h_{a})$ is a continuous mapping. Therefore, taking into account the budget constraint we obtain as an optimal solution of \eqref{eq:example_var_red}:
	\[ a_n(x)= \min\{ a \in [a^*_n, \VaR_{\alpha_n}(Y_{n+1})]: \ \pi_{R,n}(h_{a}) \leq x^+ \}. \]
	Consequently, an optimal reinsurance policy $\pi=(d_0^*,\dots,d_{N-1}^*)$ is given by $d_n^*(x)= h_{ a_n(x)}$. I.e.\ it is an optimal policy to buy a layer reinsurance treaty in each period where the single parameter is chosen as close to the optimal parameter of the corresponding problem without constraint as the current surplus allows. For a low surplus, this means that the insurer should invest all capital in reinsurance to mitigate future insurance claims rather than saving capital to pay for them himself.
\end{example}

\begin{remark}
	Since the value functions are decreasing, translation invariance of the risk measure implies that for a sufficiently high initial capital the recursive cost of capital is non-positive at each stage, cf.\ Theorem \ref{thm:finite}. In this case, we get an upper bound for the probability of ruin before the planning horizon $N$ in the setting of Example \ref{ex:var}. Fix a sufficiently large $x >0$, let $f_n^*$ be the optimal reinsurance contract at time $n$ in state $x$ and
	\[ U_n=f_n^*(Y_{n+1}) + \pi_{R,n}(f_n^*) -z_{n+1} -x  + \beta J_{n+1}\big( x + z_{n+1} - f_n^*(Y_{n+1}) - \pi_{R,n}(f_n^*)\big). \]
	Then it follows
	\begin{align*}
		0 \geq J_n(x)= \VaR_{\alpha_n}(U_n) = \inf\{t \in \R: \P(U_n>t) \leq 1-\alpha_n \}
	\end{align*}
	and especially $\P(U_n>0) \leq 1-\alpha_n$. This is the probability of a ruin in period $[n,n+1)$ if the future cost of capital is taken into account. Now, Bonferroni's inequality yields that the probability of a ruin before the planning horizon $N$ is bounded by $\sum_{n=0}^{N-1} (1-\alpha_n)$ under the optimal reinsurance policy if the recursive cost of capital is non-positive at each stage. Note that we are discussing here an imputed ruin taking into account future costs. Omitting the future cost part in the Bellman equation at each stage, one finds that the probability of a ruin in the classical sense (negative surplus) under the so-obtained optimal policy is also bounded by $\sum_{n=0}^{N-1} (1-\alpha_n)$. 
\end{remark}

As a second example we consider the dynamic reinsurance problem with general risk measures but without a budget constraint. This turns out to simplify the problem significantly.

\begin{example}\label{ex:no_constraint}
	Let the model be stationary since we consider both finite and infinite planning horizons. In case of no budget constraint $D_n(x)=\F$ for all $x \in \R$, the dynamic optimization problems \eqref{eq:opt_crit_finite} and \eqref{eq:infinite_value_function} reduce to static problems and there is a constant optimal action. This can be seen by backward induction. At time $N-1$, the Bellman equation reads due to the translation invariance of $\rho$
	\[ J_{N-1}(x) = \min_{f \in \F}   \rho \Big(  f(Y) - Z\Big) + \pi_R(f) -x ,  \]
	i.e.\ the minimization does not depend on the state of the surplus process $x$. Therefore, the value function is of the form
	\[ J_{N-1}(x)= c - x\]
	with a constant $c = \min_{f \in \F}   \rho (  f(Y) - Z) + \pi_R(f)$ and the optimal decision rule $d_{N-1}^*$ is constant
	\[ d_{N-1}^*(x) = \argmin_{f \in \F}   \rho (  f(Y) - Z) + \pi_R(f) =:f^*, \qquad x \in \R. \]
	Proceeding to the previous time step, we get due to translation invariance and positive homogeneity of $\rho$
	\begin{align*}
	J_{N-2}(x) &= \min_{f \in \F}   \rho \Big(  f(Y)+\pi_R(f) - Z -x + \beta J_{N-1}\big(x + Z - f(Y) - \pi_R(f) \big)  \Big) \\
	&= \min_{f \in \F}  \rho \Big(  f(Y) + \pi_R(f) - Z-x  + \beta \big( c + f(Y) + \pi_R(f) - Z -x \big) \Big)\\
	&= \min_{f \in \F}  (1+ \beta) \Big(\rho\big( f(Y) - Z \big) + \pi_R(f)\Big) -(1+\beta) x   + \beta c.
	\end{align*}
	Again, the minimization does not depend on $x$, the value function is given by
	\[ J_{N-2}(x) = (1+2\beta) c  -(1+\beta) x \]
	and the optimal decision rule is $d_{N-2}^*\equiv f^*$. Continuing with the induction, one finds that the value functions are affine and structurally related to the bounding functions
	\begin{align*}
	J_n(x) &= c \sum_{k=0}^{N-n-1} (k+1) \beta^k - x \sum_{k=0}^{N-n-1} \beta^k, & x \in \R,\ n=0,\dots, N-1,\\
	J(x) &= \frac{c}{(1-\beta)^2} - \frac{x}{1-\beta}, & x \in \R.
	\end{align*}
	Moreover, there is a retained loss function $f^* \in \F$ which is optimal at each point in time independently from the state of the surplus process. It can be determined by solving the classical static optimal reinsurance problem
	\begin{align}\label{eq:no_constraint}
		\min_{f \in \F} \  \rho\big(f(Y)-Z\big) + \pi_R(f).
	\end{align}
	In order to prove this, it remains to verify the induction step. Due to translation invariance and positive homogeneity of $\rho$ it follows
	\begin{align*}
	J_n(x) &= \min_{f \in \F}   \rho \Big(  f(Y)+\pi_R(f) - Z -x + \beta J_{n+1}\big(x + Z - f(Y) - \pi_R(f) \big)  \Big)\\
	&= \min_{f \in \F}   \rho \Big(  f(Y)+\pi_R(f) - Z -x +  c \beta\sum_{k=0}^{N-n-2} (k+1) \beta^k \\
	& \phantom{= \min_{f \in \F}   \rho \Big(}\ + \beta \big( f(Y)+\pi_R(f) - Z -x \big)\sum_{k=0}^{N-n-2}\beta^k  \Big)\\
	&= \min_{f \in \F} \Big(\rho\big( f(Y) - Z \big) + \pi_R(f)\Big)\sum_{k=0}^{N-n-1}\beta^k  +  c \beta\sum_{k=0}^{N-n-2} (k+1) \beta^k - x \sum_{k=0}^{N-n-1}\beta^k\\
	&= c \left(\sum_{k=0}^{N-n-1}\beta^k + \beta\sum_{k=0}^{N-n-2} (k+1) \beta^k  \right) -x \sum_{k=0}^{N-n-1}\beta^k\\
	&= c \sum_{k=0}^{N-n-1} (k+1) \beta^k - x \sum_{k=0}^{N-n-1} \beta^k.
	\end{align*} 
	If the premium income is deterministic, the optimal retained loss function in \eqref{eq:no_constraint} is known for several risk measures. Under Value-at-Risk, layer reinsurance contracts are optimal as we have shown in Example \ref{ex:var} and this remains true under Expected Shortfall with the possibility of a degeneration to a stop-loss treaty, see \cite{ChiTan2013}. More complicated multi-layer treaties are known to be optimal under general distortion risk measures and Wang premium principles, see e.g.\ \citet{Cui2013} and \cite{Zhuang2016}. 
\end{example}

In examples \ref{ex:var} and \ref{ex:no_constraint} the optimal reinsurance policy turned out to be myopic. The following example shows that in general this is not the case if there is a budget constraint.

\begin{example}\label{ex:ES}
	Consider a stationary model with $N=2, \ Y \sim U(0,1), \ Z\equiv z \in \R_+$, the expected premium principle $\pi_R=(1+\theta) \E$ and Expected Shortfall as risk measure $\rho=\ES_{\alpha}$. Realistically chosen parameters satisfy $\frac{1}{1-\alpha} \geq 1+\theta$. In the terminal period, we have to solve the Bellman equation
	\begin{align}\label{eq:ES_1}
		\min_{f \in D(x)} \ES_{\alpha}(f(Y)) + (1+\theta) \E[Y-f(Y)] -z-x.
	\end{align}
	Here, we used Theorem \ref{thm:finite} and the translation invariance of $\rho$.  We show now that an optimal retained loss function can be found in the class of \emph{stop-loss} treaties $f(x) = \min\{x,a\}$, $a \in [0,1]$. For every $f \in \F$ we can choose $a_f \in [0,1]$ such that 
	\begin{align}\label{eq:ES_2}
		\min\{Y,a_f\} \leq_{cx} f(Y).
	\end{align}
	To see this, note that the mapping $[0,1] \to \R_+, \ a \mapsto \E[\min\{Y,a\}]$ is continuous by dominated convergence and $\E[\min\{Y,0\}] \leq \E[f(Y)] \leq \E[\min\{Y,1\}]$. Hence, by the intermediate value theorem there is an $a_f \in [0,1]$ such that $\E[f(Y)] = \E[\min\{Y,a_f\}]$. Let us compare the survival functions
	\begin{align*}
		S_{\min\{Y,a_f\}}(y)&=\P(\min\{Y,a_f\}>y)=\P(Y>y)\1\{a_f>y\},\\
		S_{f(Y)}(y)&= \P(f(Y)>y)\leq \P(Y>y).
	\end{align*}
	The inequality holds since $f \leq \id_{\R_+}$. Hence, we have $S_{\min\{Y,a_f\}}(y) \geq S_{f(Y)}(y)$ for $y<a_f$ and $S_{\min\{Y,a_f\}}(y) \leq S_{f(Y)}(y)$ for $y\geq a_f$. The cut criterion 1.5.17 in \citet{MuellerStoyan2002} implies $\min\{Y,a_f\} \leq_{icx} f(Y)$ and due to the equality in expectation follows \eqref{eq:ES_2}, cf.\ Theorem 1.5.3 in \citet{MuellerStoyan2002}. Note that Expected Shortfall preserves the convex order $\leq_{cx}$, see Theorem 4.3 in \citet{BaeuerleMueller2006}, and we have equality in expectation. Thus, $y \mapsto \min\{y,a_f\}$ is weakly better that $f$ w.r.t.\ the objective function \eqref{eq:ES_1} and satisfies the budget constraint if $f$ does. Therefore, the problem is reduced to finding the optimal parameter of a stop-loss treaty.
	
	Interchanging quantiles with the increasing and continuous function $y \mapsto \min\{y,a\}$ as in Example \ref{ex:var}, we can reformulate \eqref{eq:ES_1} to
	\begin{align}\label{eq:ES_3}
		\min_{a \in [0,1]} \frac{1}{1-\alpha} \int_{\alpha}^1 \min\{u,a\} \dif u + (1+\theta) \left( \frac12 - \int_0^1 \min\{  u,a\} \dif u \right) - z - x
	\end{align}
	with the constraint  
	\begin{align}\label{eq:ES_4}
		(1+\theta) \left( \frac12 - \int_0^1 \min\{  u,a\} \dif u \right) = \frac{1+\theta}{2} (1-a)^2 \leq x^+. 
	\end{align}
	If we consider \eqref{eq:ES_3} without constraint, it follows from $\frac{1}{1-\alpha} \geq 1+\theta$ that the optimal $a$ must be smaller than $\alpha$ which reduces the objective function to 
	\begin{align*}
		\min_{a \in [0,\alpha]} a + \frac{1+\theta}{2} (1-a)^2 - z - x
	\end{align*}
	with minimizer $\hat a=\min\left\{\frac{\theta}{1+\theta}, \alpha \right\}$. Since the left hand side of the constraint \eqref{eq:ES_4} is decreasing in $a$, the inequality can be transformed to
	\[ a \geq \left(1-\sqrt{\frac{2x^+}{1+\theta}}\right)^+. \]
	Consequently, the optimal parameter of the problem with constraint \eqref{eq:ES_3} is 
	\[ a^*(x) = \max\left\{\min\left\{\frac{\theta}{1+\theta}, \alpha\right\}, \left(1-\sqrt{\frac{2x^+}{1+\theta}}\right)^+ \right\}. \]
	I.e.\ the value function at time $n=1$ is
	\[ J_1(x) = a^*(x) + \frac{1+\theta}{2} (1-a^*(x))^2 - z - x \]
	resulting in a structurally different optimization problem at time $n=0$ and a non-myopic optimal reinsurance policy.
\end{example}

\noindent{\bf Acknowledgments.} The author would like to thank Nicole B\"{a}uerle for inspiring discussions and valuable comments.

\bibliographystyle{apalike}
\renewcommand{\bibfont}{\small}
\bibliography{literature_dynamic_reinsurance}

\end{document}